\documentclass[a4paper,UKenglish,cleveref, autoref, thm-restate, authorcolumns]{lipics-v2021}

\usepackage{cite}
\usepackage{amsmath,amssymb,amsfonts}
\usepackage{algorithmic}
\usepackage{graphicx}
\usepackage{textcomp}
\usepackage{xcolor}
\usepackage{stmaryrd}

\usepackage[ruled,vlined]{algorithm2e}

\usepackage{multirow}

\usepackage{tikz}
\usepackage{quantikz}
\usepackage{subcaption}

\usepackage{url}

\usepackage{pgfplots}

\usepackage{thmtools}
\usepackage{thm-restate}

\usetikzlibrary{patterns}

\ifdefined \VersionWithComments
\usepackage{marginnote}
\newcommand{\marginX}{\marginnote{\huge{\quad\quad\textbf{!}\quad\quad}}}
\newcommand{\lsw}[1]{\mbox{}{\color{orange}\marginX{}\textbf{[Shang-Wei}: #1]}}
\newcommand{\cyr}[1]{\mbox{}{\color{green!50!black}\marginX{}\textbf{[Yean-Ru}: #1]}}
\newcommand{\yjq}[1]{\mbox{}{\color{teal}\marginX{}\textbf{[Ji-Qing}: #1]}}
\newcommand{\zhe}[1]{\mbox{}{\color{violet}\marginX{}\textbf{[Zhe}: #1]}}
\newcommand{\david}[1]{\mbox{}{\color{purple}\marginX{}\textbf{[David}: #1]}}
\newcommand{\instructions}[1]{{\color{red}\marginX{}\textbf{[Instructions: ``#1'']}}}
\newcommand{\reviewer}[2]{\mbox{}{\color{red}\marginX{}\textbf{[Reviewer #1}: ``#2'']}}
\newcommand{\todo}[1]{\mbox{}{\color{blue}{\marginX{}\textbf{TODO}\ifx#1\\\else:\ \fi #1}}} 

\else
\newcommand{\instructions}[1]{}
\newcommand{\cyr}[1]{}
\newcommand{\lsw}[1]{}
\newcommand{\yjq}[1]{}
\newcommand{\david}[1]{}
\newcommand{\zhe}[1]{}
\newcommand{\reviewer}[2]{}
\newcommand{\todo}[1]{}
\fi

\usepackage{thmtools}
\usepackage{thm-restate}

\nolinenumbers

\title{\ourwork: A Quantum SAT Solver Based on Entanglement and Equivalence Checking} 

\author{Shang-Wei LIN}{InfoComm Technology (ICT) Cluster, Singapore Institute of Technology (SIT), Singapore \and \url{https://shangweilin.github.io/} }{shangwei.lin@singaporetech.edu.sg}{https://orcid.org/0000-0002-9726-3434}{}

\author{Ji-Qing YAN}{Department of Electrical Engineering, National Cheng Kung University, Taiwan \and \url{} }{n26124939@gs.ncku.edu.tw}{}{}

\author{Yean-Ru CHEN}{Department of Electrical Engineering, National Cheng Kung University, Taiwan \and \url{} }{chenyr@mail.ncku.edu.tw}{}{}

\author{Zhe HOU}{InfoComm Technology (ICT) Cluster, Singapore Institute of Technology (SIT), Singapore \and \url{} }{z.hou@griffith.edu.au}{}{}

\author{David Sanán}{InfoComm Technology (ICT) Cluster, Singapore Institute of Technology (SIT), Singapore \and \url{} }{sanan.baena@gmail.com}{}{}

\authorrunning{S.-W. Lin et al.} 

\Copyright{Shang-Wei LIN and Yean-Ru CHEN} 



\category{} 

\relatedversion{} 

\newcommand{\ourwork}[0]{\texttt{Quantangle-SAT}}

\begin{document}

\maketitle

\begin{abstract}
Satisfiability (SAT) is a central problem in computer science, and advances in SAT-solving algorithms have a far-reaching impact across many fields. Recent works have proposed quantum SAT solvers based on Grover’s algorithm, a quantum search technique. However, Grover-based approaches face a key limitation: they typically require prior knowledge of the number of satisfying assignments of the target Boolean formula. This information is unavailable in most practical settings. Quantum counting can be used to estimate this quantity, but it incurs a computational overhead that is several orders of magnitude higher than Grover search. In this paper, we propose a novel quantum SAT solver based on entanglement and equivalence checking. Our method does not assume prior knowledge of the number of solutions and is computationally more efficient than quantum counting. Although the worst-case time complexity is inevitably exponential, we prove that the expected time complexity of our approach is only constant time $O(1)$ over random Boolean functions. Experimental results also support our theoretical claim.
\end{abstract}

\section{Introduction} \label{sec:Intro}

Satisfiability (SAT) is a fundamental problem in computer science. Finding efficient algorithms for SAT solving has significant impacts on many areas of computer science and beyond. 
Quantum computing generalizes classical computing from binary bits to quantum bits (qubits), which could represent both 0's and 1's simultaneously in {\em superpositions}. This characteristic provides the ability to execute all the possible computational paths simultaneously, known as quantum parallelism. Furthermore, qubits can become {\em entangled} with each other, which is a strictly quantum mechanical phenomenon with no classical analogue and provides a computing resource enabling quantum computers to achieve \emph{quantum supremacy} over their classical counterparts. These properties of quantum computing lead to substantial speed-up compared to certain classical computing algorithms.

Recently, quantum SAT solvers have been proposed based on Grover's algorithm~\cite{Grover96}, a quantum searching technique. However, Grover-based quantum SAT solvers~\cite{FS19,LWC24} have a fundamental limitation: the number of solutions of the target Boolean formula must be known a priori, which is 
usually not the case. Although quantum counting~\cite{Mos98,BHMT02,AR21} can be used to estimate the number of solutions to fulfill the prerequisite of Grover's algorithm, its computational complexity is higher than Grover's algorithm by multiple orders of magnitude. This makes Grover-based quantum SAT solvers impractical in applications.

\begin{figure}[tb]
\begin{subfigure}[b]{0.12\linewidth}
\scalebox{0.6}{
\begin{quantikz}[row sep={7mm,between origins}, column sep=5mm, font=\Large]
\lstick{$\ket{x}$} & \ctrl{1}\gategroup[2,steps=1,style={fill=yellow!40},background]{} & \qw \\
\lstick{$\ket{F}$} & \targ{} & \qw
\end{quantikz}
}
\vspace{5mm}
\caption{Simple}\label{subfig:Ex_Simple}
\end{subfigure}
~
\begin{subfigure}[b]{0.18\linewidth}
\scalebox{0.6}{
\begin{quantikz}[row sep={7mm,between origins}, column sep=3mm, font=\Large]
\lstick{$\ket{x}$} & \qw & \ctrl{2}\gategroup[3,steps=3,style={fill=yellow!40},background]{} & & \ctrl{2} & \qw & \qw \\
\lstick{$\ket{F}$} & \qw & & \targ{} &  & \qw & \qw \\
\lstick{$\ket{0}$} & \qw & \targ{} & \ctrl{-1} & \targ{} & \qw & \qw
\end{quantikz}
}
\vspace{5mm}
\caption{Good} \label{subfig:Ex_Good}
\end{subfigure}
~
\begin{subfigure}[b]{0.62\linewidth}
\includegraphics[width=\linewidth]{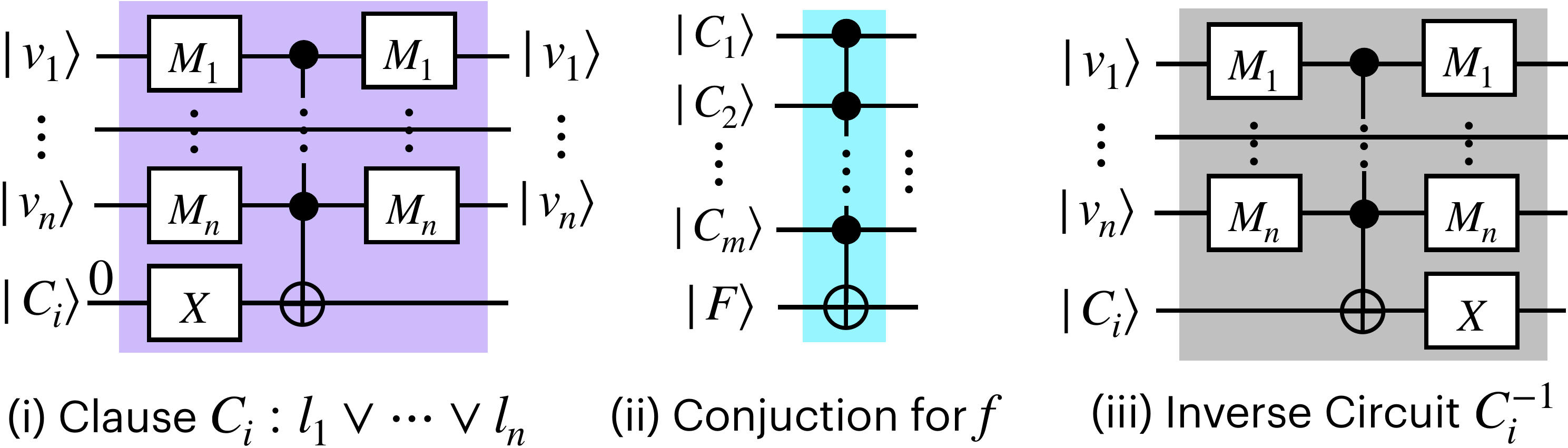}
\caption{Quantum Circuit Construction for A CNF Formula} \label{subfig:ClassicConstruction}
\end{subfigure}
\caption{Quantum Circuits} \label{fig:Running_Example}
\end{figure}

In this work, we propose a novel quantum SAT solver that does not require prior knowledge of the number of solutions.
We tackle the problem from a different angle.  
The intuition behind our approach is neat:
Given a Boolean formula $f$, we determine whether $f$ is satisfiable by comparing it with 
\emph{falsity}
$f_{\bot}$, which has no solutions. If $f$ is equivalent to $f_\bot$, then $f$ is unsatisfiable; otherwise, $f$ is satisfiable. To utilize 
quantum parallelism, we convert both formulas, $f$ and $f_\bot$, into quantum circuits, and then compare them 
using
equivalence checking of quantum circuits~\cite{SW22}. However, 
the original
quantum equivalence checking~\cite{SW22} 
has two limitations:
\begin{enumerate}
\item The widths (the number of qubits) of the two quantum circuits have to be exactly the same. However, when constructing a quantum circuit, 
different designers may use different numbers of ancillary qubits, resulting in different widths, even for equivalent quantum circuits.
For example, the two circuits, as shown in Figs.~\ref{subfig:Ex_Simple} and ~\ref{subfig:Ex_Good}, implement the same Boolean function $f: (x)$ consisting of one clause with only one positive literal $x$. Since these two circuits have different widths, the work~\cite{SW22} is not applicable here. 

\item Even when two quantum circuits $U_1$ and $U_2$ have the same width $n$, the work~\cite{SW22} further requires to extend these two circuits into their corresponding larger versions $U_1'$ and $U_2'$,
each of which has width $2n$,
 for precise distance calculation. How this extension is done will be introduced in Section~\ref{subsec:EquivalenceChecking}. This additional overhead becomes substantial when $n$ (the number of variables) is large,
which is common when scaling up SAT solving.
\end{enumerate}

Inspired by~\cite{SW22}, we propose \ourwork, a quantum SAT solver based on entanglement and equivalence checking. 
We not only address the above two limitations in quantum equivalence checking but also go beyond to use the improved checking method to solve SAT problems. Our contributions can be summarized as follows:
\begin{itemize}
    \item We solve the problem of inconsistent width when applying equivalence checking on quantum circuits with different widths (c.f.~Section~\ref{subsec:ancillary-qubits}).

    \item We develop a general scheme of quantum circuit construction for arbitrary Boolean formulas such that the overhead required in the second limitation of~\cite{SW22}
    is eliminated (c.f.~Section~\ref{subsec:TightBound}).

    \item With the two limitations addressed, our approach determines whether a Boolean formula $f$ is satisfiable by comparing $f$ with $f_\bot$ based on equivalence checking on their corresponding quantum circuits $\Ddot{U}_f$ and $\Ddot{U}_{f_\bot}$, obtained by 
    our new
    quantum circuit construction scheme (c.f.~Section~\ref{subsec:QuantumSATSolving}). Furthermore, we prove that the \emph{expected} time complexity of our approach is constant time (c.f. Section~\ref{sec:Discussion}).
\end{itemize}

To the best of our knowledge, our approach, \ourwork, is the first quantum SAT solver based on entanglement and equivalence checking. Compared to Grover-based SAT solvers, \ourwork~does not require
the number of solutions of the target formula a priori and has lower complexity than quantum counting. Notice that our approach is not a competitor of classical non-quantum SAT solvers. Instead, \ourwork~is complementary to them. When the SAT problem cannot be solved by classical non-quantum SAT solvers (e.g., the number of Boolean variables is too huge), \ourwork~provides another different approach.
We have formally proved the correctness of our approach. All the proofs are provided in the appendix due to the page limit.


\section{Preliminary} \label{sec:Preliminary}

We assume that the readers have basic knowledge of linear algebra and quantum computing, e.g., the {\em tensor product} operation, primitive quantum gates (e.g., $X$, $Z$, $H$, etc), and {\em quantum entanglement}. Here, we briefly establish the notation used throughout this paper. We use the {\em ket} notation $\ket{\cdot}$ to denote the (column) vector  and the {\em bra} notation $\bra{\cdot}$ to denote its conjugate transpose (row) vector. 
Given two vectors $\ket{v_1}$ and $\ket{v_2}$, we use $\langle v_1 | v_2 \rangle$ to denote their inner product, $|v_1\rangle \langle v_2|$ for their outer product, and $\ket{v_1} \otimes \ket{v_2}$ for their tensor product, which may be abbreviated as $\ket{v_1} \ket{v_2}$, or even $\ket{v_1 v_2}$ for simplicity.
In this work, we 
use
\emph{projective measurement}, as defined in Definition~\ref{def:ProjectiveMeasurement}.

\begin{definition} \label{def:ProjectiveMeasurement}
A projective measurement is described by an observable, $M$, a Hermitian operator on the system being observed. The observable $M$ has a spectral decomposition $M = \sum_a a \cdot P_a$, where $P_a$ is the projector onto the eigenspace of $M$ with eigenvalue $a$.
\end{definition}

The possible outcomes of the measurement correspond to the eigenvalues, $a$, of the
observable. We use the notation $\widetilde{M}$ to denote the measurement result of $M$, and use $M_{\ket{\phi}}$ to denote the operation of measurement on the state $\ket{\phi}$. Thus, $\widetilde{M_{\ket{\phi}}} = a$ means that the result of applying measurement $M$ on state $\ket{\phi}$ is $a$.
After measurement, the probability of getting result $a$ is given by $\mathbb{P}(a) = \bra{\phi} P_a \ket{\phi}$.
Given that outcome $a$ occurred, the state of the quantum system immediately after the measurement is $\frac{P_a \ket{\phi}}{\sqrt{\mathbb{P}(a)}}$.
A nice property of projective measurement $M$ is that the \emph{expected value} of the measurement result on a quantum state $\ket{\phi}$, denoted by $\mathbb{E}(M_{\ket{\phi}})$, can be easily calculated as $    \mathbb{E}(M_{\ket{\phi}}) = \bra{\phi} \! M \! \ket{\phi}$.

\subsection{Quantum Circuit Construction} \label{subsec:QC_Construction}


In this work, we consider SAT formulas in conjunctive normal form (CNF) over a set of Boolean variables $V$, whose syntax is as follows:
\[
\begin{array}{lllll}
    F \, \simeq \,  C_1 \wedge C_2 \wedge \cdots \wedge C_m & \hspace{10mm} &
    C \, \simeq \,  l_1 \vee l_2 \vee \cdots \vee l_n & \hspace{10mm} &
    l \, \simeq \, v \mid \overline{v}
\end{array}
\]

A formula $F$ is a conjunction of $m$ clauses $C_1, C_2, \ldots, C_m$, and each clause $C_i$ is a disjunction of $n$ literals $l_1, l_2, \ldots, l_n$, where $m, n \in \mathbb{N}$. A literal $l_j$ could be a Boolean variable $v$ and called a {\em positive} literal, or the negation of a Boolean variable $\overline v$ and called a {\em negative} literal. 
A formula can be viewed as a function $F: \{0,1\}^{|V|} \mapsto \{0, 1\}$ mapping an input vector $\vec{v} \in \{0,1\}^{|V|}$ to true/false (0/1), where $|V|$ denotes the cardinality of $V$. A formula $F$ is \emph{satisfiable} if there exists some $\vec{v} \in \{0,1\}^{|V|}$ such that $F(\vec{v}) = 1$, and we call such $\vec{v}$ a {\em solution} to $F$.  A formula $F$ is \emph{unsatisfiable} if it does not have any solution. 

The construction of the quantum circuit for a CNF formula $F$ follows the bottom-up approach~\cite{FS19,LCWC23}. The circuit for each clause $C_i$ is constructed first, and then all the clauses are conjuncted together. Fig.~\ref{subfig:ClassicConstruction}(i) shows how to construct the circuit for each clause $C_i: l_1 \vee l_2 \vee \ldots \vee l_n$, where the $M_j$ gate depends on literal $l_j$ for $j \in \{1,2, \ldots, n\}$. If $l_j$ is positive, $M_j$ is the $X$ gate, while if $l_j$ is negative, $M_j$ is the $I$ gate. The qubit $\ket{C_i}$ represents the truth value of clause $C_i$. 
Once the quantum circuits for all the $m$ clauses are constructed, they are conjuncted by a CNOT gate ($m$-qubit Toffoli gate, to be more precise) to form the circuit for $F$, as shown in Fig.~\ref{subfig:ClassicConstruction}{(ii)}, where $\ket{F}$ represents the truth value of formula $F$, which is controlled by $\ket{C_i}$ for all $i \in \{1,2, \ldots, m\}$.

\subsection{Bell Expression and Bell Experiment}

In 1964, 
Bell derived the very famous Bell inequality~\cite{Bell64} to test whether the predictions of quantum mechanics (especially those involving entangled particles) could be explained by hidden variables that obey local realism. Quantum experiments, called the Bell experiments~\cite{CHSH69,ADR82}, were performed, and the results were interpreted based on the inequality. It turned out that Bell inequality was violated, showing that quantum correlations could not be explained by Einstein's theory of local realism. After that, Bell inequality was generalized as Bell expression~\cite{SATW17} for arbitrary number of measurements and
outcomes.

Here, we quickly introduce the Bell experiment based on Bell expression.
It is performed by two parties (Alice and Bob) and repeated for $s$ rounds, where $s \in \mathbb{N}$. For each round, Alice and Bob prepare the \emph{maximally entangled state} $\ket{\Phi_d}$, as shown in Eq.~\ref{eq:MaxEntangledState}.
\begin{equation} \label{eq:MaxEntangledState}
\ket{\Phi_d} = \frac{1}{\sqrt{d}} \sum_{i = 0}^{d-1} \ket{ii}, \quad \mbox{where} \quad d = 2^n \quad \mbox{and} \quad n \in \mathbb{N}.
\end{equation}
They hold their corresponding $n$-qubit subsystems $\ket{i}$, respectively, as shown in Fig.~\ref{subfig:BellExperiment}. Notice that the two subsystems, $\ket{i}$ held by Alice and $\ket{i}$ (denoted identically on purpose) held by Bob, are entangled. 
Then, they perform local measurements on their subsystems $\ket{i}$, respectively.
In each round, they select one observable for measurement among the $m$ predefined observables, where $m \in \mathbb{N}$.
After $s$ rounds, they collect the measurement outcome statistics data $\mathbb{P}(ab | xy)$, where $x,y \in \{1, 2, \ldots, m \}$ are the labels of the observables used for the local measurements, and $a,b \in \{ 0, 1, \ldots, d-1 \}$ are the corresponding outcomes.

\begin{figure}[tb]
\begin{subfigure}[b]{0.29\linewidth}
\includegraphics[width=\linewidth]{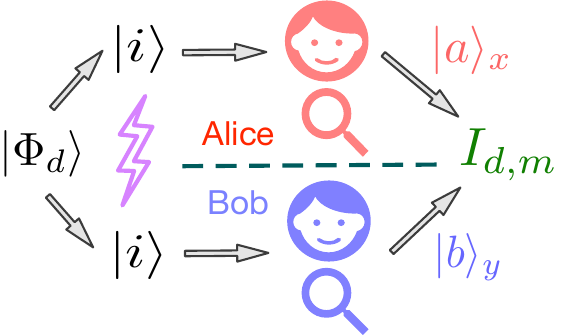}
\caption{Bell Experiment} \label{subfig:BellExperiment}
\end{subfigure}
~
\begin{subfigure}[b]{0.43\linewidth}
\includegraphics[width=\linewidth]{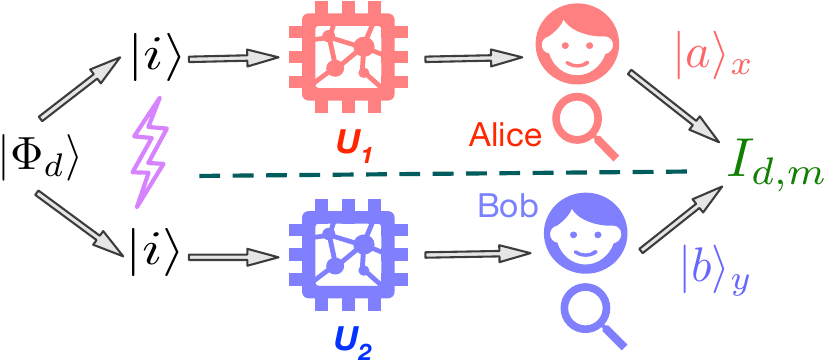}
\caption{Quantum Equivalence Checking} \label{subfig:QEC}
\end{subfigure}
~
\begin{subfigure}[b]{0.23\linewidth}
\includegraphics[width=\linewidth]{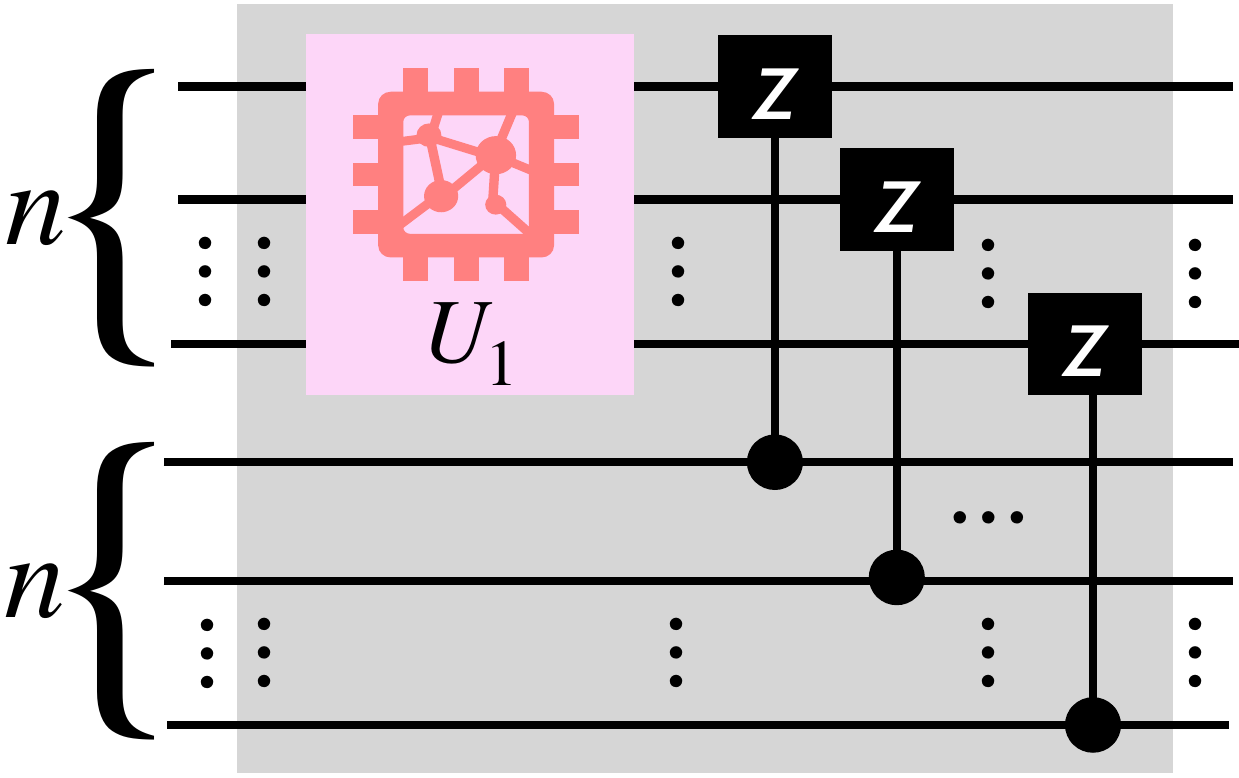}
\vspace{0.5mm}
\caption{Extented $U_1'$} \label{subfig:CircuitExtension}
\end{subfigure}
\caption{Bell Experiment and Quantum Equivalence Checking} \label{fig:BellQEC}
\end{figure}


When performing the local measurement labelled by $x \in \{1, 2, \ldots, m \}$, Alice measures an observable with the eigenvectors $\ket{a}_x$, as shown in the left of Eq.~\ref{eq:Eigenvectors}, for $a \in \{0, 1, \ldots, d -1 \}$, where $\omega = e^{\frac{2\pi \mathbf{i}}{d}}$ and $\alpha_x = (x - \frac{1}{2}) / m$.
\begin{equation} \label{eq:Eigenvectors}
\ket{a}_x = \frac{1}{\sqrt{d}} \sum_{k=0}^{d-1} \omega^{ k \cdot(a - \alpha_x)} |k\rangle \quad \quad \quad \mbox{vs.} \quad \quad \quad
\ket{b}_y = \frac{1}{\sqrt{d}} \sum_{k=0}^{d-1} \omega^{ - k \cdot (b - \beta_y)} |k\rangle
\end{equation}
To perform the local measurement labelled by $y$, Bob measures an observable with the eigenvectors $\ket{b}_y$, as shown in the right of Eq.~\ref{eq:Eigenvectors}, for $b \in \{0, 1, \ldots, d -1 \}$, where $\beta_y = y / m$.
After the local measurements are performed, the results are interpreted by Bell expression, formulated in Definition~\ref{def:BellExp}. Notice that Bell expression is defined for an arbitrary quantum state, not just for the maximally entangled state $\ket{\Phi_d}$.

\begin{definition}[Bell Expression~\cite{SATW17}] \label{def:BellExp}
Bell expression for measuring an arbitrary quantum state $\ket{\phi}$
, denoted by $I_{d,m}(\ket{\phi})$ 
is defined as $I_{d,m}(\ket{\phi}) = \sum_{i = 1}^{m} \sum_{l=1}^{d-1} \langle \phi| (A_i^l \otimes \overline{B}_i^l) \ket{\phi}$, where $A_i^l = \sum_{a=0}^{d-1} \omega^{ a \cdot l} |a\rangle_{ii}\langle a|$ and $\overline{B}_i^l = (A_i^l)^*$. Here, $\ket{a}_i$ is the eigenvector defined in Eq.~\ref{eq:Eigenvectors}, and $|a\rangle_{ii}\langle a|$ is the outer product of $\ket{a}_i$.
\end{definition}

Interestingly, Lemma~\ref{lm:MaxViolation} shows that the value of Bell expression has an upper bound, and this upper bound is achieved when the state being measured is the maximally entangled state $\ket{\Phi_d}$. Notice that, in the rest of this paper, if the state $\ket{\phi}$ to be measured is clear without ambiguity, we will sometimes omit $\ket{\phi}$ and just use $I_{d,m}$ to denote $I_{d,m}(\ket{\phi})$ to save space.

\begin{restatable}[Maximal Violation of Bell Expression~\cite{SATW17}]{lm}{MaxViolation} \label{lm:MaxViolation} ~\\
Given a quantum state $\ket{\phi}$, we have $I_{d,m}(\ket{\phi}) \leq m(d-1)$. In addition, the upper bound is exactly achieved when $\ket{\phi}$ is the maximally entangled state $\ket{\Phi_d}$.
\end{restatable}


\subsection{Equivalence Checking of Quantum Circuits} \label{subsec:EquivalenceChecking}

The contributions of Bell expression/experiment are not only theoretical but also practical. The work of \cite{SW22} showed that Bell expression can be adopted to check whether two quantum circuits are equivalent.
Let $U_1$ and $U_2$ be two quantum circuits, each of which has $n$ qubits. The equivalence checking for $U_1$ and $U_2$ is fundamentally a Bell experiment by two parties: Alice for $U_1$ and Bob for $U_2$, respectively. The experiment will be repeated for $s$ rounds, where $s \in \mathbb{N}$.
The whole equivalence checking process is shown in Fig.~\ref{subfig:QEC}. One can see that the setting is almost identical to a Bell experiment except that the two subsystems are passed through $U_1$ and $U_2$, respectively, before being measured by the same observables (defined in Definition~\ref{def:BellExp}). Lemma~\ref{lm:CircuitEquiv} shows that the value of Bell expression will achieve the upper bound if and only if $U_1$ and $U_2$ are equivalent.


\begin{restatable}[Circuit Equivalence~\cite{SW22}]{lm}{CircuitEquiv} \label{lm:CircuitEquiv}
Let $U_1$ and $U_2$ be two quantum circuits. We have
$I_{d,m}((U_1 \otimes U_2) \ket{\Phi_d}) = m (d - 1) \iff U_1 = U_2$ up to a global phase.
\end{restatable}

If $U_1 \neq U_2$, how do we know how different they are? One useful indicator is the distance between them. The distance between $U_1$ and $U_2$, denoted by $D(U_1, U_2)$ is defined as $D(U_1, U_2) = \sqrt{1 -  \left| \frac{1}{d} \mathsf{Tr}(U_1^T U_2) \right|^2}$, where $\mathsf{Tr}(M)$ denotes the sum of diagonal elements in the square matrix $M$. Lemma~\ref{lm:DRange} shows that the upper and lower bounds of $D(U_1, U_2)$ can be obtained based on the Bell expression. 

\begin{restatable}[Matrix Distance Range~\cite{SW22}]{lm}{MatrixDistance} \label{lm:DRange}
Let $\ket{\phi} = (U_1 \otimes U_2) \ket{\Phi_d}$ and $\Omega = I_{d,m}(\ket{\phi})$, then we have $\sqrt{1 - \frac{\Omega + m}{md}} \leq D(U_1, U_2) \leq \sqrt{1 - \frac{\Omega - m (d - 2)}{m}}$.
\end{restatable}

In practice, knowing the upper and lower bounds of $D(U_1,U_2)$ may not be good enough --- we may need the exact value of $D(U_1,U_2)$. Lemma~\ref{lm:IDM_Prime} shows that, if we spend $n$ more qubits extending $U_1$ into $U_1'$ based on the construction in Fig.~\ref{subfig:CircuitExtension}, and similarly extending $U_2$ into $U_2'$, we can obtain the exact value of $D(U_1, U_2)$ by equivalence checking of $U_1'$ and $U_2'$. Notice that both $U_1'$ and $U_2'$ have $2n$ qubits, and thus the value of $d$ in Lemma~\ref{lm:IDM_Prime} is $2^{2n}$, not $2^n$.

\begin{restatable}[Extended Distance~\cite{SW22}]{lm}{ExtendedDistance} \label{lm:IDM_Prime} 
Let $\ket{\phi'} = (U_1' \otimes U_2') \ket{\Phi_d}$ and $d = 2^{2n}$. We have $D(U_1', U_2') = D(U_1, U_2) = \sqrt{1 - \frac{I_{d,m}(\ket{\phi'}) + m}{md}} = \sqrt{1 - \widehat{I_{d,m}}(\ket{\phi'})}$.

\end{restatable}

If we look closely into $I_{d,m}$, we will find that the number of terms in it is $m(d-1)$, where $d = 2^n$ grows exponentially with the number of qubits $n$. This makes calculating $I_{d,m}$ computationally expensive in practice. Thus, the work~\cite{SW22} proposed to use the \emph{normalized} Bell expression $\widehat{I_{d,m}}$. Lemma~\ref{lm:NormBellExp} defines it and highlights its property.

\begin{restatable}[Normalized Bell Expression~\cite{SW22}]{lm}{NormBellExp} \label{lm:NormBellExp} 
Given a state $\ket{\phi}$ to be measured, we define $\widehat{I_{d,m}}(\ket{\phi}) =  (I_{d,m}(\ket{\phi})+ m) / md$ as the normalized Bell expression such that $0 \leq  \widehat{I_{d,m}}(\ket{\phi}) \leq 1$.
It can be further rewritten as follows, where $A_i^1$ is defined in Definition~\ref{def:BellExp} for $i \in \{1, \ldots, m\}$, $A^1_{m+1} = A^1_1 + 1$, and $B_i^1 = (A_i^1)^*$ ; $\alpha_k = \frac{1}{2d} \cdot \tan(\frac{\pi}{2m}) \cdot \cot\left(\frac{\pi}{d} (k + \frac{1}{2m}) \right)$ .
\[
    \widehat{I_{d,m}}(\ket{\phi}) = \frac{1}{m} \sum_{k=0}^{d-1} \sum_{i=1}^{m} \alpha_k 
    \left[
        \mathbb{P}\left( \widetilde{(A^1_i)_{\ket{\phi}}} = \widetilde{(B^1_i)_{\ket{\phi}}} + k \right) + \mathbb{P} \left( \widetilde{(B^1_i)_{\ket{\phi}}} = \widetilde{(A^1_{i+1})_{\ket{\phi}}} + k \right)
    \right]
\]

\end{restatable}

Lemma~\ref{lm:NormBellExp} suggests that the value of $\widehat{I_{d,m}}$ can be approximated by $\mathtt{QEC}$~\cite{SW22}, a Monte Carlo sampling process, as shown in Algorithm~\ref{alg:QEC} listed in Appendix~\ref{apx:QEC}. 
Its prototype is $\mathtt{QEC}(\ket{\phi}, n, C_1, C_2, m, s)$ and returns the approximated $\widehat{I_{d,m}}(\ket{\phi})$, where $\ket{\phi}$ is the input state, $n$ is the width of two quantum circuits $C_1$ and $C_2$, $m$ is the number of measurement types, and $s$ is the number of samplings. To get $D(U_1,U_2)$, we need to extend the two circuits into $U_1'$ and $U_2'$ and obtain $\widehat{I_{d,m}}((U_1' \otimes U_2')\ket{\Phi_d})$ by applying $\mathtt{QEC}\left(\ket{\Phi_d}, 2n, U_1', U_2', m,s\right)$. Then, we can calculate $D(U_1,U_2)$ based on Lemma~\ref{lm:IDM_Prime}. 



\section{Our \ourwork~Algorithm}
\label{sec:Methodology}

\begin{figure}[tb]
\centering
\begin{subfigure}[b]{0.3\linewidth}
\includegraphics[width=0.88\linewidth]{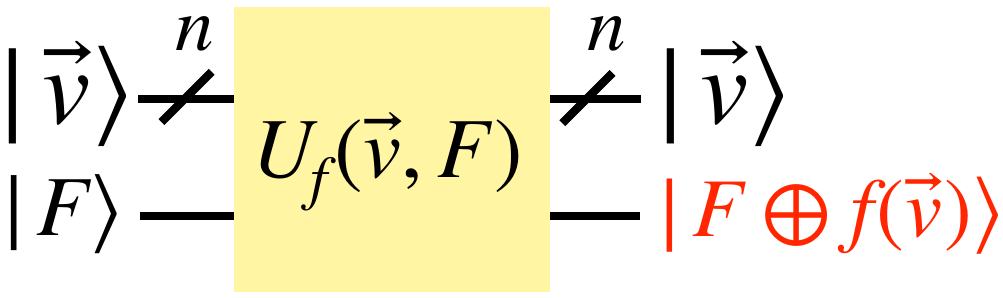}
\vspace{5mm}
\caption{Specification} \label{subfig:Uf_Spec}
\end{subfigure}
~
\hspace{7mm}
~
\begin{subfigure}[b]{0.6\linewidth}
\includegraphics[width=0.83\linewidth]{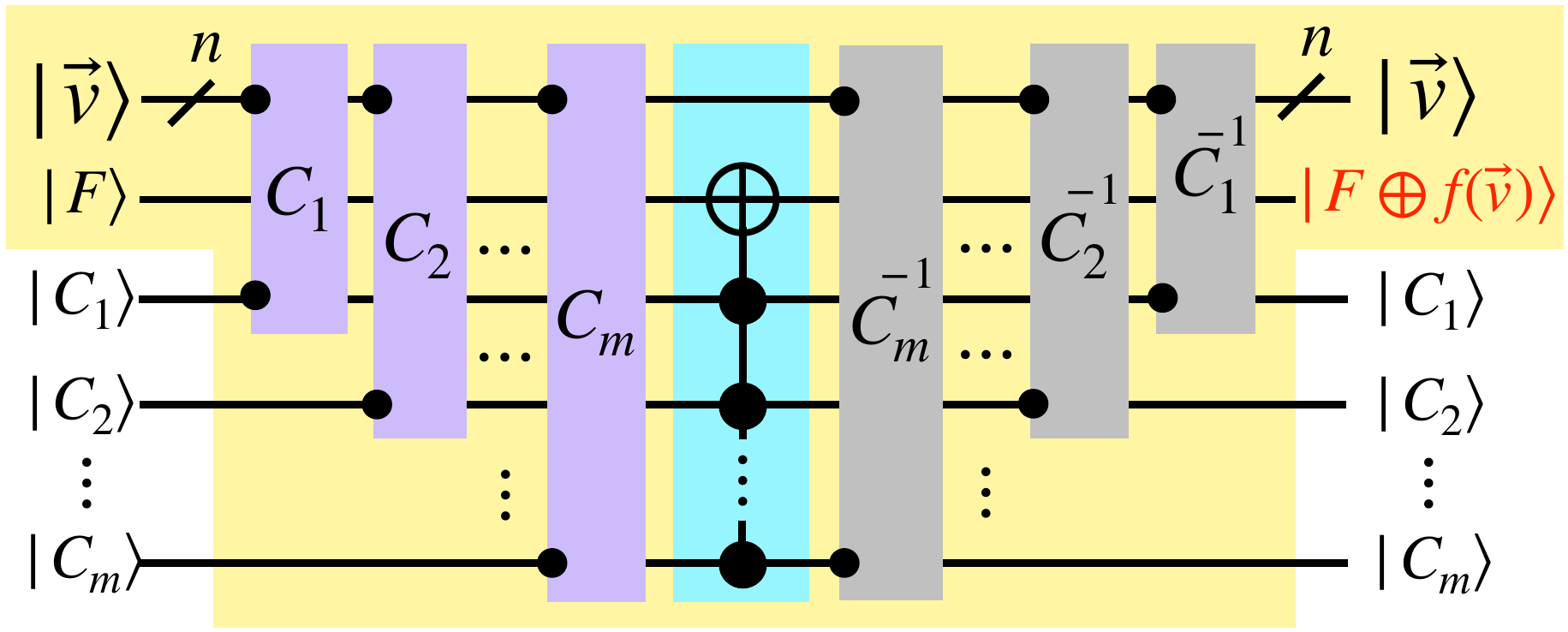}
\caption{Implementation} \label{subfig:Uf_Impl}
\end{subfigure}
\caption{Quantum Circuit $U_f$ for a CNF Formula $f$} \label{fig:UfConstruction}
\end{figure}

\subsection{Solving The Inconsistent Width Problem} \label{subsec:ancillary-qubits}

The $\mathtt{QEC}$ algorithm~\cite{SW22} can be applied to check whether two quantum circuits $U_1$ and $U_2$ are equivalent, only when they have the same circuit width (i.e., they have the same number of qubits). 
In this subsection, we will demonstrate how to solve the inconsistent width problem. First of all, let us figure out the minimal number of qubits required to implement a quantum circuit for a formula $f$ over a set of Boolean variables $V = \{v_1, v_2, \ldots, v_n \}$. Since a quantum circuit has to be reversible, in addition to the $n$ input qubits, we need one more qubit for the output value, denoted by $\ket{F}$. That is, the minimal number of qubits required is $n+1$. The standard mathematical construction is shown in Fig.~\ref{subfig:Uf_Spec}, where $\vec{v}$ denotes the vector $(v_1, v_2, \ldots, v_n)$. However, this mathematical construction is ideal in theory but infeasible in practice, where ancillary qubits are usually required. For example, the qubits $\ket{C_1}, \ldots, \ket{C_m}$ in Fig.~\ref{subfig:ClassicConstruction}{(ii)} are ancillaries and indicate which clauses are satisfied.
For ease of illustration, let us use $U_f(\vec{v}, F)$ to denote the mathematical construction, while using $U_f(\vec{v}, F, \vec{a})$ to denote the practical construction with ancillary qubits $\vec{a}$.
Before we move on, let us recall Lemma~\ref{lm:InvariantDirectSumDecomposition}, the theorem of \emph{Invariant Direct Sum Decomposition}, in linear algebra:



\begin{restatable}[Invariant Direct Sum Decomposition~\cite{Hefferon20}]{lm}{InvariantDirectSumDecomposition} \label{lm:InvariantDirectSumDecomposition} ~\\
Let $\mathcal{V} = W_1 \oplus W_2 \oplus \cdots \oplus W_k$ be a vector space that is the direct sum of its subspaces. Given a linear transformation $T: \mathcal{V} \mapsto \mathcal{V}$, if each subspace $W_i$ is a $T$-invariant subspace of $\mathcal{V}$, i.e., $\,T(W_i) = W_i$ for all $i \in \{ 1, 2, \ldots, k \}$, then there exists an ordered basis $\beta = \{ \beta_1, \beta_2, \ldots, \beta_k \}$ for $\mathcal{V}$, where $\beta_i$ is the basis of subspace $W_i$, such that $T$ can be represented w.r.t. $\beta$ as the diagonal matrix $[T]_\beta$, as shown in Fig.~\ref{subfig:TBeta}, where $M_i = [T_{W_i}]_{\beta_i}$. 
\end{restatable}
Let $\vec{a}$ denote the vector of $(a_1, a_2, \ldots, a_l)$ consisting of $l$ ancillary qubits. It is obvious that $\mathcal{V}_{\ket{\vec{v}} \ket{F} \ket{\vec{a}}} = \bigoplus W_{(a_1, a_2, \ldots, a_l)}$. That is, each configuration of $(a_1, a_2, \ldots, a_l)$ defines one subspace, and $\mathcal{V}_{\ket{\vec{v}} \ket{F} \ket{\vec{a}}}$ is the direct sum of all the $2^l$ subspaces. Lemma~\ref{lm:InvariantDirectSumDecomposition} suggests that, if we can make each subspace $W_{(a_1, a_2, \ldots, a_l)}$ to be a $U_f(\vec{v},F,\vec{a})$-invariant subspace of $\mathcal{V}_{\ket{\vec{v}} \ket{F} \ket{\vec{a}}}$, then $U_f(\vec{v},F,\vec{a})$ will be equivalent to $U_f(\vec{v},F)$ in one configuration of $(a_1, a_2, \ldots, a_l)$. The way to achieve this goal is to restore the ancillary qubits to their initial states. 

Fig.~\ref{subfig:Uf_Impl} shows the overall scheme of our quantum circuit construction that restores all the ancillary qubits for clauses back to their initial states. The process to restore each ancillary qubit $\ket{C_i}$ is done by performing the inverse operation $C_i^{-1}$, as shown in Fig.~\ref{subfig:ClassicConstruction}{(iii)}. 

Let us take the ``simple'' and ``good'' circuits in Figs.~\ref{subfig:Ex_Simple} and~\ref{subfig:Ex_Good} for illustration. The ``good'' circuit follows our construction in Fig.~\ref{subfig:Uf_Impl}, and Fig.~\ref{subfig:T_Inv_Ex_Good} shows the component-by-component correspondence. It can be easily checked that $\mathcal{V}_{\ket{x} \ket{F} \ket{C}} = W_{\ket{0}} \oplus W_{\ket{1}}$, where $W_{\ket{0}}$ is spanned by the basis $\{ \ket{000}, \ket{010}, \ket{100}, \ket{110} \}$, while $W_{\ket{1}}$ by $\{ \ket{001}, \ket{011}, \ket{101}, \ket{111} \}$. Basically, $W_{\ket{0}}$ is the subspace where the ancillary qubit $\ket{C}$ is $\ket{0}$, while $W_{\ket{1}}$ for $\ket{1}$. In our circuit construction, the initial state of $\ket{C}$ is $\ket{0}$, and once the result of the formula is determined by the blue component, $\ket{C}$ is restored back to its initial state by performing the gray component. In such a setting, $W_{\ket{0}}$ and $W_{\ket{1}}$ are $U_f(x,F,C)$-invariant subspaces of $\mathcal{V}_{\ket{x} \ket{F} \ket{C}}$. By Lemma~\ref{lm:InvariantDirectSumDecomposition}, $U_f(x,F,0)$, the yellow sub-matrix in Fig.~\ref{subfig:T_Inv_Ex_Good}, is equivalent to the matrix of $U_f(x,F)$ in Fig.~\ref{subfig:T_Inv_Ex_Simple}. Thus, although $U_f(x,F,0)$ and $U_f(x,F)$ have different widths, $\mathtt{QEC}$ is still applicable to them, but focusing on the input/output qubits only, i.e., $\ket{x}$ and $\ket{F}$.


\begin{figure}[tb]
\centering
\begin{subfigure}[b]{0.17\linewidth}
\includegraphics[width=0.73\linewidth]{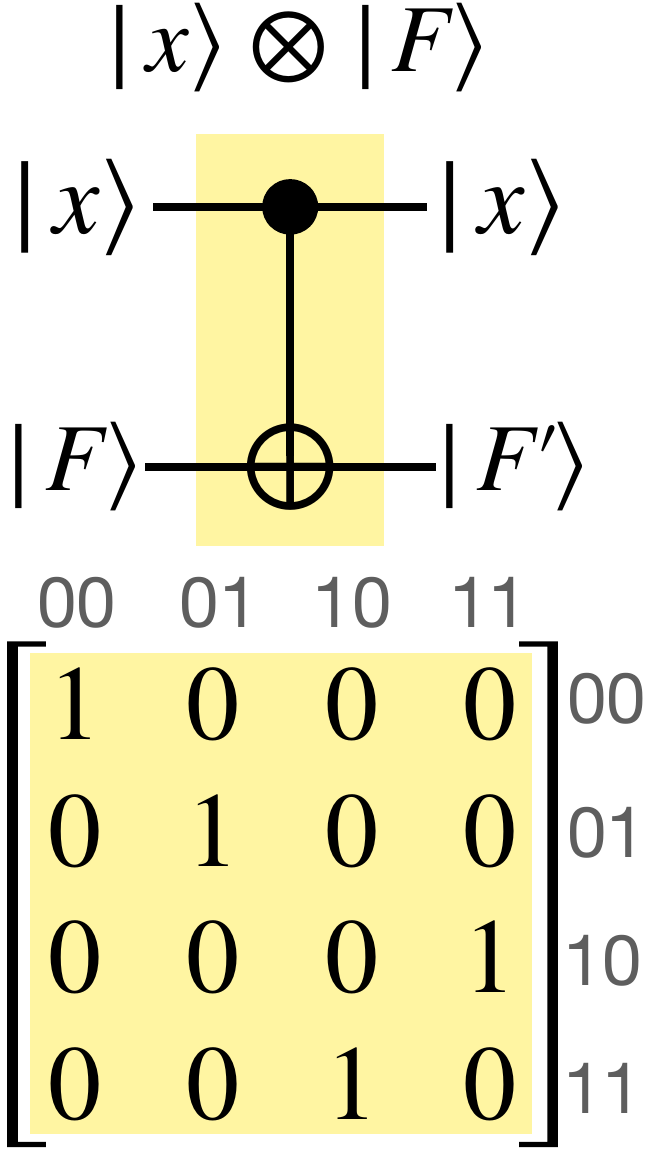}
\caption{Simple} \label{subfig:T_Inv_Ex_Simple}
\end{subfigure}
~
\begin{subfigure}[b]{0.49\linewidth}
\includegraphics[width=0.91\linewidth]{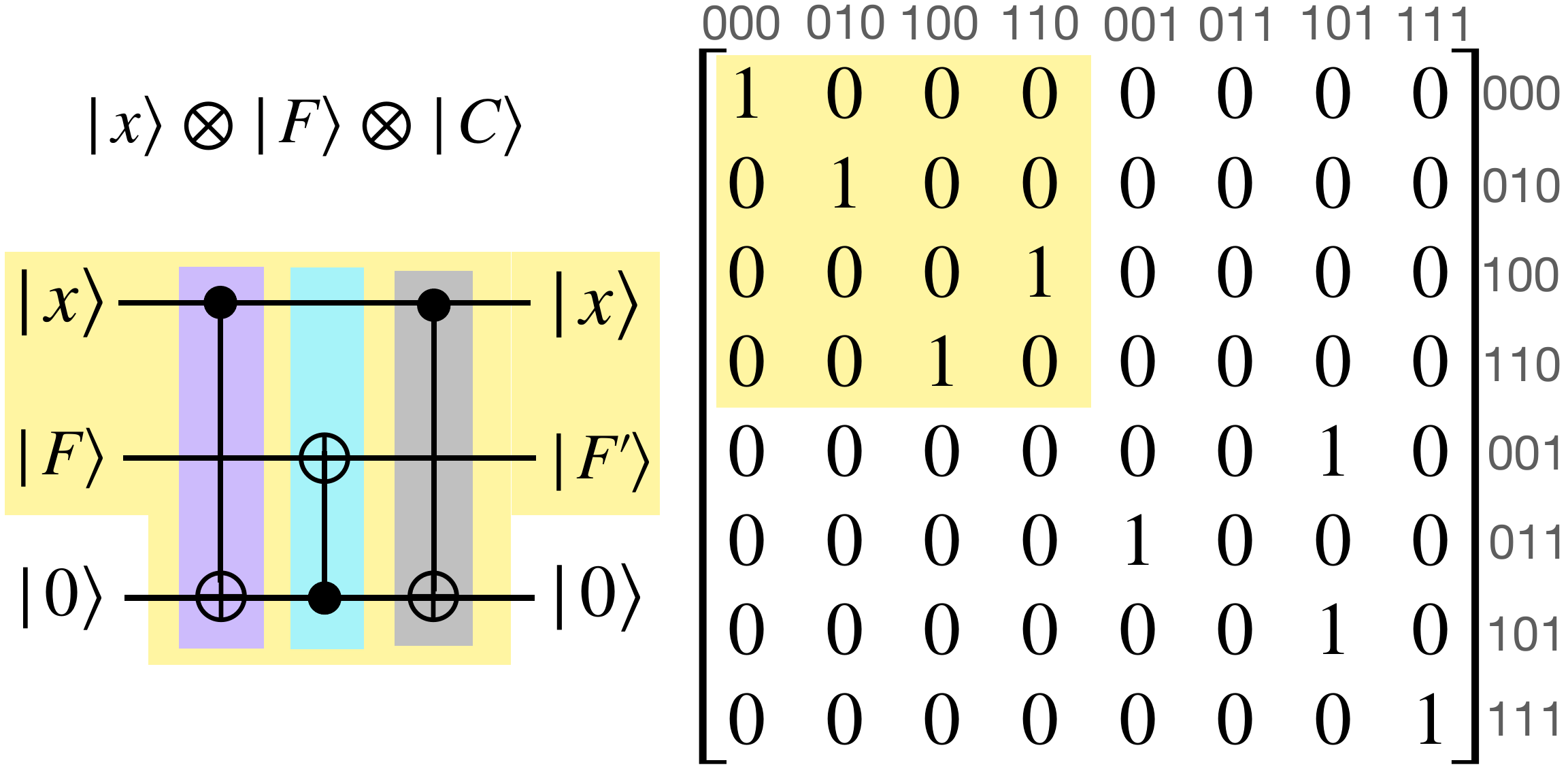}
\caption{Good} \label{subfig:T_Inv_Ex_Good}
\end{subfigure}
~
\begin{subfigure}[b]{0.29\linewidth}
\scalebox{0.83}{
$
\left[
\begin{array}{c|c|c|c}
 \;M_1\; & 0 & \cdots & 0 \\ \hline
 0 & \;M_2\; & \cdots & 0 \\ \hline
 \vdots & \vdots & \;\cdots\; & \vdots \\ \hline
 0 & 0 & \cdots & \;M_k\;
\end{array}
\right]
$
}
\vspace{5mm}
\caption{$[T]_\beta$} \label{subfig:TBeta}
\end{subfigure}
\caption{Invariant Direct Sum Decomposition} \label{fig:TInv_Subspace_Decomp_Example}
\end{figure}







    
    
    
    

    
    
            


\subsection{Collapsing The Upper and Lower Bounds of $D(U_1, U_2)$} \label{subsec:TightBound}

Lemma~\ref{lm:DRange} gives the upper and lower bounds of $D(U_1,U_2)$. Now, we present our finding, Theorem~\ref{thm:TightBound}, which collapses the upper bound and lower bound of $D(U_1,U_2)$ to a single value such that when applying the $\mathtt{QEC}$ algorithm on $U_1$ and $U_2$, the circuit extension process, as shown in Fig.~\ref{subfig:CircuitExtension}, is no longer required. 

\begin{restatable}[]{thm}{TightBound} \label{thm:TightBound}
Suppose $U_1$ and $U_2$ are quantum operators with $n$ qubits. Let $\mathcal{S}$ be the subspace spanned by the basis $\{ \ket{j}\ket{j} \}_d$, where $d = 2^n$ and $j \in \{0, 1, \ldots, d  -1 \}$. One can quickly check that $\ket{\Phi_d} \in \mathcal{S}$. If $\ket{\phi} = (U_1 \otimes U_2) \ket{\Phi_d}$ still belongs to $\mathcal{S}$, i.e., $\ket{\phi} \in \mathcal{S}$, then $D(U_1, U_2) = \sqrt{1 - \widehat{I_{d,m}}(\ket{\phi})}$.
\end{restatable}

Intuitively, Theorem~\ref{thm:TightBound} says that if $U_1 \otimes U_2$ maps $\ket{\Phi_d}$ to a state $\ket{\phi}$ such that $\ket{\phi}$ always remains in the subspace $\mathcal{S}$, then $D(U_1, U_2)$ would be exactly $\sqrt{1 - \widehat{I_{d,m}}(\ket{\phi})}$. Now, we show how to construct quantum circuits $U_1$ and $U_2$ to satisfy this requirement so as to enjoy this nice property. Given a Boolean formula $f$ in CNF form consisting of $m$ clauses, the way to construct each clause $C_i$ for $i \in \{1,2,\ldots,m\}$ is still the same as shown in Fig.~\ref{subfig:ClassicConstruction}{(i)}, which implies that its inverse circuit $C_i^{-1}$ also remains the same as shown in Fig.~\ref{subfig:ClassicConstruction}{(iii)}. However, when we do the conjunction of all the clauses, we change to use the new circuit, as shown in Fig.~\ref{subfig:UfDotDot_Conj}. Basically, this operation applies the $Z$ gate on the result qubit $F$, when all the clauses are satisfied.
Putting it all together, Fig.~\ref{subfig:UfDotDot} shows the overall scheme of our improved quantum construction. From now on, we use $\Ddot{U}_f$ to denote the quantum circuit constructed based on this improved scheme for a given formula $f$. Theorem~\ref{thm:UfDotDot-1} proves that $\Ddot{U}_f (\ket{\vec{v}}\ket{F})$ does not change the state of $\ket{\vec{v}} \ket{F}$ but adds a `$-1$' phase to it, when $f(\vec{v}) = 1$ and $F$ is $1$.

\begin{restatable}[]{thm}{UfDotDot} \label{thm:UfDotDot-1}
$\Ddot{U}_f (\ket{\vec{v}}\ket{F}) = (-1)^{f(\vec{v}) \wedge F} \cdot \ket{\vec{v}}\ket{F}$
\end{restatable}

The main advantage of our improved quantum circuit construction is that,
given two formulas $f_1$ and $f_2$ over the same set of Boolean variables $V$, we can construct their corresponding $\Ddot{U}_{f_1}$ and $\Ddot{U}_{f_2}$, respectively. By Theorem~\ref{thm:UfDotDot-1}, they will not change their input states but only add a `$-1$' phase. Thus, when we apply $\Ddot{U}_{f_1} \otimes \Ddot{U}_{f_2}$ on $\ket{\Phi_d}$, the result state $\ket{\phi}$ will still remain in $\mathcal{S}$. Further by Theorem~\ref{thm:TightBound}, we do not need to extend $\Ddot{U}_{f_1}$ and $\Ddot{U}_{f_2}$, when we apply $\mathtt{QEC}$ on them.

\begin{figure}[tb]
\centering
\begin{subfigure}[b]{0.23\linewidth}
\includegraphics[width=0.49\linewidth]{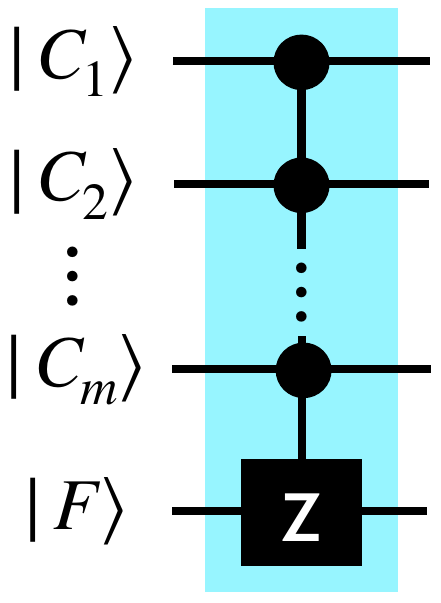}
\caption{Conjunction for $F$} \label{subfig:UfDotDot_Conj}
\end{subfigure}
~
\hspace{3mm}
~
\begin{subfigure}[b]{0.7\linewidth}
\includegraphics[width=0.7\linewidth]{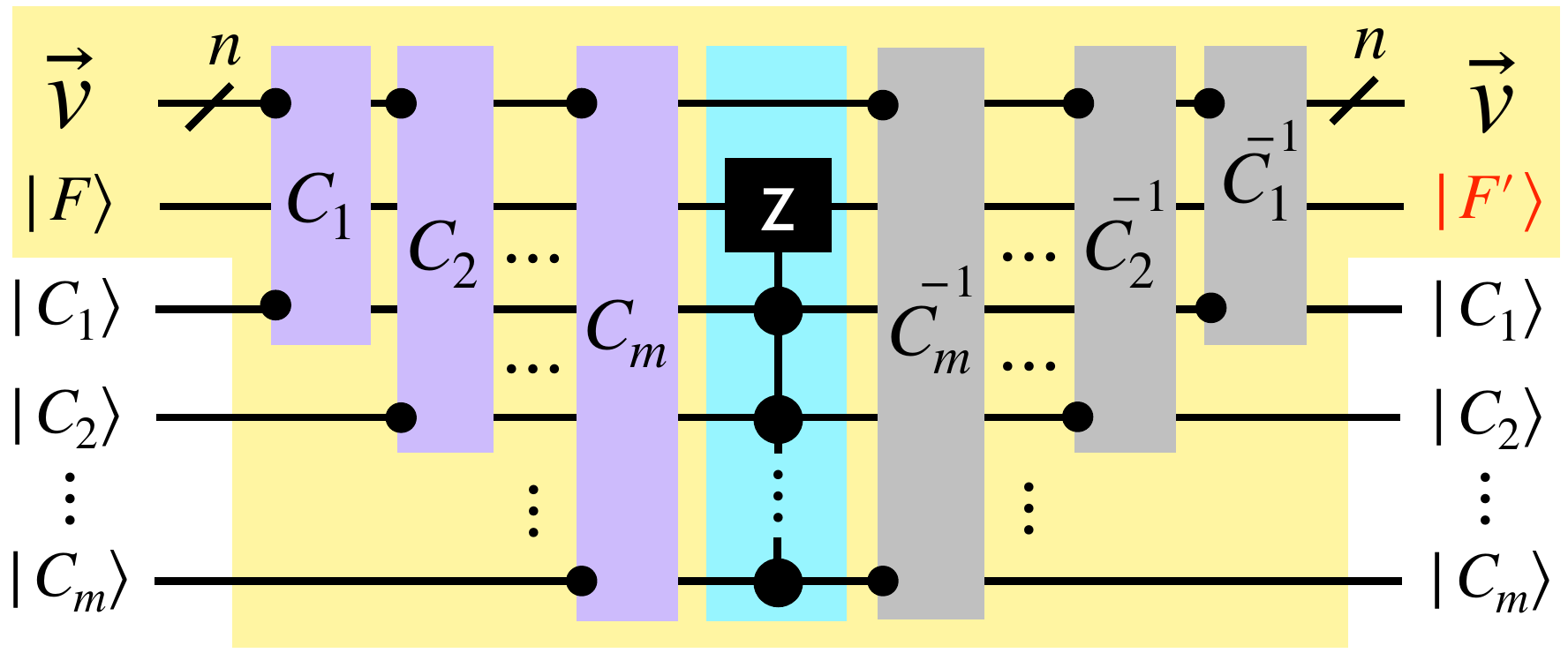}
\caption{Overall Construction}  \label{subfig:UfDotDot}
\end{subfigure}
\caption{Improved quantum circuit $\Ddot{U}_f$ for a CNF formula $f$} \label{fig:UfDotDot}
\end{figure}

\subsection{Quantum SAT Solving} \label{subsec:QuantumSATSolving}

Now, we present how quantum SAT solving can be done based on 
the above
building blocks. Given a formula $f$, our idea is to compare $f$ with 
falsity
$f_\bot$. If $f$ is equivalent to $f_\bot$, then $f$ is unsatisfiable; otherwise, $f$ is satisfiable. This comparison can be done by performing the $\mathtt{QEC}$ algorithm on $\Ddot{U}_f$ and $\Ddot{U}_{f_\bot}$. 
Theorem~\ref{thm:FalseIdentity} 
shows
that $\Ddot{U}_{f_\bot}$ is 
simply
the identity matrix.

\begin{restatable}[]{thm}{FalseIdentity} \label{thm:FalseIdentity}
Let $f$ be a Boolean formula consisting of $n$ variables. If formula $f$ is unsatisfiable, then $\Ddot{U}_f$ is the $2^{n+1}$ by $2^{n+1}$ identity matrix.
\end{restatable}

Therefore,
when comparing $\Ddot{U}_f$ and $\Ddot{U}_{f_\bot}$, we do not even have to explicitly construct $\Ddot{U}_{f_\bot}$. Furthermore, with our improved circuit construction scheme, the trace of $\Ddot{U}_f$ would be $2^{n+1} - 2K$, where $K$ is the number of solutions to formula $f$, which is proved in Lemma~\ref{lm:FormulaTrace}.

\begin{restatable}[]{lm}{FormulaTrace} \label{lm:FormulaTrace}
Given a Boolean formula $f$ with $n$ variables, if formula $f$ has $K$ solutions, then $\mathsf{Tr}(\Ddot{U}_f) = 2^{n+1} - 2K$.
\end{restatable}

With this nice property, we can further figure out the relation between $\widehat{I_{d,m}}$ and the solution ratio $r$ of formula $f$, where $r = \frac{K}{2^n}$ and $K$ is the number of solutions to $f$. Theorem~\ref{thm:IDMvsRatio} proves that $\widehat{I_{d,m}} = (1 - r)^2$.

\begin{restatable}[]{thm}{IDMvsRatio} \label{thm:IDMvsRatio}
$\widehat{I_{d,m}} = (1 - r)^2$, where $r = \frac{K}{2^n}$.
\end{restatable}

Theorem~\ref{thm:IDMvsRatio} tells us that if formula $f$ has no solutions, i.e., $r = 0$, then the value of $\widehat{I_{d,m}}$ would be exactly $1$; otherwise, it would be strictly less than $1$. 
Since the value of $\widehat{I_{d,m}}$ depends on $K$, we use the notation $\widehat{I_{d,m,K}}$ to indicate the $\widehat{I_{d,m}}$ value for the formula $f$ with $K$ solutions.
For example, given a formula $f$ over $4$ Boolean variables, Fig.~\ref{subfig:decision_boundary-A} shows the corresponding $\widehat{I_{d,m,K}}$ values, marked as blue dots, for different $K \in \{ 0,1,2, \ldots, 16\}$. Interestingly, Lemma~\ref{lm:IDM_RandomVar} shows that $\widehat{I_{d,m,K}}$ can be considered as the expected value of the random variable $X_{d,m,K}$, where $\alpha_k$, $A_i^1$, $B_i^1$ are defined in Lemma~\ref{lm:NormBellExp}.

\begin{restatable}[]{lm}{IDMRandomVar} \label{lm:IDM_RandomVar}
Let $\ket{\phi}$ be a quantum state and $X_{d,m,K}$ be the random variable over the values of $x_k = 2 \alpha_k$ with probability $\mathbb{P}(x_k) = \frac{1 }{m} \sum_{i=1}^{m} \mathbb{P}\left( \widetilde{(A^1_i)_{\ket{\phi}}} = \widetilde{(B^1_i)_{\ket{\phi}}} + k \right)$ for $k \in \{0,1,\ldots, d-1\}$. We have $\widehat{I_{d,m,K}}(\ket{\phi})  = \mathbb{E}(X_{d,m,K}) = \sum_{k=0}^{d-1} x_k \cdot \mathbb{P}(x_k)$.
\end{restatable}

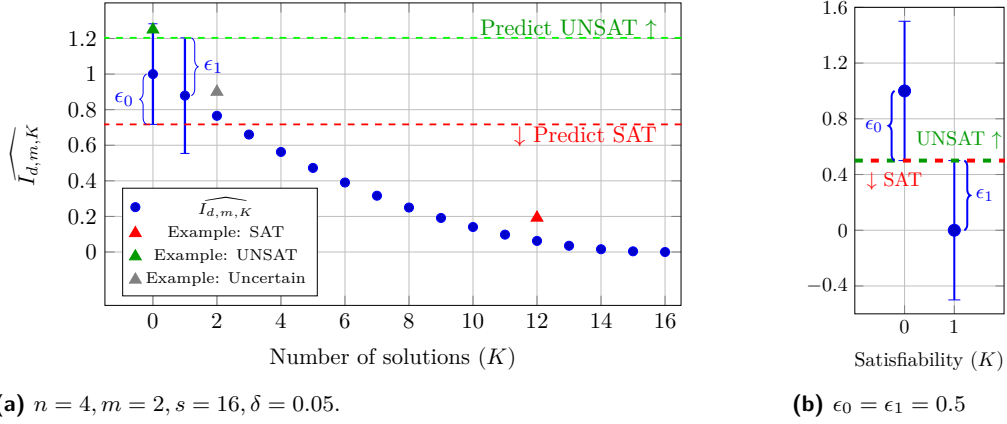
\begin{figure}[tb]
\centering
    \begin{subfigure}[b]{0.75\textwidth} 
        \scalebox{0.85}{
            \begin{tikzpicture}
                \begin{axis}[
                    width=\textwidth,
                    height=0.6\textwidth,
                    xlabel={Number of solutions ($K$)},
                    ylabel={$\widehat{I_{d,m,K}}$},
                    xmin=-1.5, xmax=16.5,
                    ymin=-0.3, ymax=1.4,
                    xtick={0,2,4,6,8,10,12,14,16},
                    ytick={0, 0.2, 0.4, 0.6, 0.8, 1.0, 1.2},
                    grid=major,
                    legend pos=south west,
                    legend style={font=\scriptsize}
                ]

                \addplot+[
                    color=blue,
                    only marks,
                    mark=*,
                    error bars/.cd,
                    y dir=both,
                    y explicit,
                    error bar style={line width=1pt, solid}
                ] table [x=x, y=y, y error=err, row sep=crcr] { 
                    x   y           err \\
                    0   1.0         0.282758 \\ 
                    1   0.878906    0.324816 \\ 
                    2   0.765625    0.0 \\      
                    3   0.660156    0.0 \\
                    4   0.5625      0.0 \\
                    5   0.472656    0.0 \\
                    6   0.390625    0.0 \\
                    7   0.316406    0.0 \\
                    8   0.25        0.0 \\
                    9   0.191406    0.0 \\
                    10  0.140625    0.0 \\
                    11  0.097656    0.0 \\
                    12  0.0625      0.0 \\
                    13  0.035156    0.0 \\
                    14  0.015625    0.0 \\
                    15  0.003906    0.0 \\
                    16  0.0         0.0 \\
                };
                \addlegendentry{$\widehat{I_{d,m,K}}$};

                \addplot[mark=triangle*, mark size=3pt, red, only marks] coordinates {(12, 0.1925)};
                \addlegendentry{Example: SAT};
                \addplot[mark=triangle*, mark size=3pt, green!60!black, only marks] coordinates {(0, 1.25)};
                \addlegendentry{Example: UNSAT};
                \addplot[mark=triangle*, mark size=3pt, gray, only marks] coordinates {(2, 0.9)};
                \addlegendentry{Example: Uncertain};

                \addplot[color=red, dashed, thick, domain=-2:17] {0.7172};
                \addplot[color=green, dashed, thick, domain=-2:17] {1.2037};
                
                \node[anchor=east, color=red] at (axis cs:16, 0.65) {$\downarrow$ Predict SAT};
                \node[anchor=east, color=green!60!black] at (axis cs:16, 1.25) {Predict UNSAT $\uparrow$};

                \draw[decorate,decoration={brace,amplitude=3pt,mirror},blue] 
                    (axis cs:-0.2, 1.0) -- (axis cs:-0.2, 0.7172) node [midway, left=2pt, font=\normalsize] {$\epsilon_0$};
                \draw[decorate,decoration={brace,amplitude=3pt},blue] 
                    (axis cs:1.2, 1.2037) -- (axis cs:1.2, 0.8789) node [midway, right=2pt, font=\normalsize] {$\epsilon_1$};

                \end{axis}
            \end{tikzpicture}
        }
        \caption{$n=4, m=2, s=16, \delta=0.05$.} \label{subfig:decision_boundary-A}
    \end{subfigure}%
    \hfill
    \begin{subfigure}[b]{0.25\textwidth}
        \scalebox{0.75}{
            \begin{tikzpicture}
                \begin{axis}[
                    width=1.2\textwidth,
                    height=2\textwidth,
                    ymin=-0.6, ymax=1.6,
                    ytick={-0.4, 0, 0.4, 0.8, 1.2, 1.6},
                    xmin=-1, xmax=2,
                    xlabel={Satisfiability ($K$)},
                    xtick={0, 1},
                    grid=major
                ]
                \addplot+[
                    color=blue,
                    only marks,
                    mark=*,
                    mark size=3pt,
                    error bars/.cd,
                    y dir=both,
                    y explicit,
                    error bar style={line width=1pt, solid}
                ] table [x=x, y=y, y error=err, row sep=crcr] {
                    x   y      err \\
                    0   1.0    0.5 \\
                    1   0.0    0.5 \\
                };
                \addplot[color=green!60!black, line width=2pt, dash pattern=on 4pt off 12pt, domain=-1:17] {0.5};
                \addplot[color=red, line width=2pt, dash pattern=on 4pt off 12pt, dash phase=8pt, domain=-1:17] {0.5};
                
                \node[anchor=east, color=red] at (axis cs:0.5, 0.35) {$\downarrow$ SAT};
                \node[anchor=east, color=green!60!black] at (axis cs:2.1, 0.65) {UNSAT $\uparrow$};

                \draw[decorate, decoration={brace, amplitude=3pt, mirror}, blue, line width=1pt] 
                    (axis cs:-0.2, 1.0) -- (axis cs:-0.2, 0.5) node [midway, left=2pt, font=\large] {$\epsilon_0$};
                \draw[decorate, decoration={brace, amplitude=3pt}, blue, line width=1pt] 
                    (axis cs:1.2, 0.5) -- (axis cs:1.2, 0.0) node [midway, right=2pt, font=\large] {$\epsilon_1$};
                \end{axis}
            \end{tikzpicture}
        }
        \caption{$\epsilon_0=\epsilon_1=0.5$}  \label{subfig:decision_boundary-B}
    \end{subfigure}
\caption{Decision boundaries for \ourwork.} \label{fig:decision_boundaries}
\end{figure}

Upon close investigation, we found 
that the standard deviation of $X_{d,m,K}$ has a closed form.
Since our goal is to determine whether a formula $f$ has a solution or not, we are interested in the very two cases where $K = 0$ and $K = 1$. Theorem~\ref{thm:StandardVar4NoSol} gives the standard deviation for the case of $K = 0$, and Theorem~\ref{thm:StandardVar4OneSol} gives the case of $K = 1$.

\begin{restatable}[]{thm}{StandardVarForNoSol} \label{thm:StandardVar4NoSol}
Let $\ket{\phi} = \sum_{q=0}^{d-1} \gamma_q \ket{q}\ket{q}$ be a quantum state. Consider the random variable $X_{d,m,K}$ whose expected value is $\widehat{I_{d,m,K}}(\ket{\phi})$. We have $\sigma(X_{d,m,0}) = \frac{1}{d} \sqrt{\frac{d^2-1}{3}} \tan\left(\frac{\pi}{2m}\right)$.
\end{restatable}

\begin{restatable}[]{thm}{StandardVarForOneSol} \label{thm:StandardVar4OneSol}
$\sigma(X_{d,m,1}) = \frac{1}{d^2} \sqrt{ \frac{d^2(d^2-1)}{3}\tan^2\left(\frac{\pi}{2m}\right) + 4(d-1)(d-2)^2}$.
\end{restatable}


The aforementioned building blocks constitute the backbone of our quantum SAT solver, \ourwork, as shown in Algorithm~\ref{alg:QSAT}. It requires, as inputs, the target formula $f$ over the set of Boolean variables $V$, the confidence level $L$ for the SAT result, the number of observables $m$ for measurement, and the number of samples $s$ to perform $\mathtt{QEC}$. The detailed explanation of Algorithm~\ref{alg:QSAT} is as follows:

\textbf{Line~$1$}. The working variables $W$, $H$, $\mathcal{F}$, and $t$ are initialized as $V$, $\emptyset$, $f$, and $0$, respectively, where $t$ indicates the index of iterations.

\textbf{Line~$3$}. The $\mathtt{QEC}$ algorithm is applied on $\Ddot{U}_f$ and $I$ to get $\mu_t$ to approximate $\widehat{I_{d,m,K}}$. Since $\mathtt{QEC}$ is based on Monte Carlo, a statistical sampling method, there could be an additive error $|\mu_t - \widehat{I_{d,m,K}}|$ introduced.
Based on the theory of \emph{normal approximation to sample mean} (Lemma~\ref{lm:Gaussian_Tail_Probability} in Appendix~\ref{sec:UsfulLemmas}), if we set up a target error range $\epsilon$ and take $s$ samples, then we have $\mathbb{P}(|\mu_t - \widehat{I_{d,m,K}}| \geq \epsilon) \leq \delta = 2 \left( 1 - \Phi\left( \frac{\epsilon \sqrt{s}}{\sigma} \right) \right)$, where $\Phi(x) = \frac{1}{\sqrt{2\pi}} \int_{-\infty}^{x} (e^{-t^2/2}) dt$ is the cumulative distribution function (CDF) of the standard normal distribution. The more samples we take or the larger the error range $\epsilon$ we allow, the lower the error rate $\delta$ becomes.
Here, $L = 1 - \delta$ is called \emph{confidence level}. If the error rate $\delta$ is small, e.g., $0.05$, then the confidence level would be high, e.g., $0.95$. In our approach, we use Lemma~\ref{lm:Gaussian_Tail_Probability} in the opposite way. We set the desired confidence level $L$ and the number of samples $s$, and ask ``what is maximal error range $\epsilon$ allowed for $L$?''
We further analyze $\epsilon$ in the following steps: 
    
\textbf{Lines~$5$--$6$}. Since our goal is to determine whether a formula $f$ is satisfiable or not, the two critical cases are $K=0$ ($f$ has no solutions) and $K=1$ ($f$ has exactly one solution). Here, after setting up the confidence level $L$ that we want, instead of using the classical error rate $(1-L)$, we use a conservative smaller $\delta = \frac{1-L}{|V| + 1}$ to control the Family-Wise Error Rate (FWER)~\cite{Benjamini02} based on Bonferroni Correction~\cite{Bonferroni36}, because $\mathtt{QEC}$ might be applied multiple (at most $|V|+1$) times. By solving for $\epsilon$ in the equation $\delta = 2 \left( 1 - \Phi\left( \frac{\epsilon \sqrt{s}}{\sigma(X_{d,m,0})} \right) \right) = \frac{1 - L}{|V|+1}$, we can obtain the maximal error range $\epsilon_0$ allowed for the case of $K=0$. Similarly, by solving for $\epsilon$ in the equation $\delta = 2 \left( 1 - \Phi\left( \frac{\epsilon \sqrt{s}}{\sigma(X_{d,m,1})} \right) \right) = \frac{1 - L}{|V|+1}$, we can obtain $\epsilon_1$ for $K = 1$.



\textbf{Line~$7$}. Remember $\widehat{I_{d,m,0}} = 1$. If formula $f$ has no solutions, i.e., $K=0$, from Lemma~\ref{lm:Gaussian_Tail_Probability}, we have $\mathbb{P}(|\mu_t - \widehat{I_{d,m,0}}| \geq \epsilon_0) \leq \delta = \frac{1 - L}{|V|+1}$. 
That is to say, if $\mu_t$ is within the range $[1 - \epsilon_0, 1 + \epsilon_0]$, we have a $(1 - \delta)$ confidence level to conclude that $K=0$. This implies that if $\mu_t$ is outside this range, we also have $(1 - \delta)$ confidence level to conclude $K \neq 0$. Thus, if $\mu_t$ is less than $1 - \epsilon_0$, we can conclude that $f$ has at least one solution at the $(1 - \frac{1-L}{|V|+1})$ confidence level. Then, the algorithm returns ``yes'' with the solution ratio $(1 - \sqrt{\mu_0})$. Let us use Fig.~\ref{subfig:decision_boundary-A} for illustration. The red dotted line indicates the threshold $(1 - \epsilon_0)$. We can conclude $f$ has solutions for any $\mu_t$ below this threshold.

\LinesNumbered
\begin{algorithm}[tb]
\SetKwInOut{Input}{input}
\SetKwInOut{Output}{output}
\SetKwComment{Comment}{$\quad //$}{}
\Input{$f$: Boolean formula over $V$; \quad $m$: number of measurements; \\ $L$: confidence level; \quad $s$: number of samples; }
\Output{(yes/no, $r$): whether $f$ is satisfiable and the estimated solution ratio $r$}
\BlankLine

$W \longleftarrow V$ ; \quad \quad \quad $H \longleftarrow \emptyset$ ; \quad \quad \quad $\mathcal{F} \longleftarrow f$ ; \quad \quad \quad $t \longleftarrow 0$ \;

    \While{$|W| \geq 0$}{
        
        $\mu_t \longleftarrow \mathtt{QEC} \left(\ket{\Phi_d}, |W| + 1, \Ddot{U}_\mathcal{F}, I, m, s\right)$ \;
        \vspace{1mm}
        $d \longleftarrow 2^{|W| + 1}$ \;
        \vspace{-1mm}


        
        $\epsilon_0 \longleftarrow$ solving for $\epsilon$ in the equation: $2 \left( 1 - \Phi\left( \frac{\epsilon \sqrt{s}}{\sigma(X_{d,m,0})} \right) \right) = \frac{1 - L}{|V|+1}$ \;
        \vspace{1mm}
        $\epsilon_1 \longleftarrow$ solving for $\epsilon$ in the equation: $2 \left( 1 - \Phi\left( \frac{\epsilon \sqrt{s}}{\sigma(X_{d,m,1})} \right) \right) = \frac{1 - L}{|V|+1}$ \;

        \lIf{$\mu_t \le 1 - \epsilon_0$}{
            \Return{(yes, $1 - \sqrt{\mu_0}$)} 
        }
        \lElseIf{$\mu_t \ge (1 - \frac{1}{2^{|W|}})^2 + \epsilon_1$}{
            \Return{(no, $0$)}
        }
        \Else{
            pick one Boolean variable $x \in W$ \;
    
            $W \longleftarrow  W \setminus \{ x \}$ ; \quad \quad $H \longleftarrow H \cup \{ x \}$ ; \quad \quad $\mathcal{F} \longleftarrow \mathcal{F}_{[H]}$ ; \quad \quad
            $t \longleftarrow t + 1$ \;
        }
    }


\caption{\ourwork}
\label{alg:QSAT}
\end{algorithm}

\textbf{Line~$8$}. Notice that $(1 - \frac{1}{2^{|W|}})^2$ is the actual $\widehat{I_{d,m,1}}$ value when $K=1$. 
If it is indeed the case, from Lemma~\ref{lm:Gaussian_Tail_Probability} again, we have $\mathbb{P}(|\mu_t - \widehat{I_{d,m,1}}| \geq \epsilon_1) \leq \delta = \frac{1 - L}{|V|+1}$.
That is to say, if $\mu_t$ is within the range $[\widehat{I_{d,m,1}} - \epsilon_1, \widehat{I_{d,m,1}} + \epsilon_1]$, we have $(1 - \delta)$ confidence level to conclude $K = 1$. This implies that if $\mu_t$ is outside this range, we also have $(1 - \delta)$ confidence level to conclude $K \neq 1$.
Thus, if $\mu_t$ is greater than $(1 - \frac{1}{2^{|W|}})^2 + \epsilon_1$, we can conclude that $f$ has no solutions at $(1 - \frac{1-L}{|V|+1})$ confidence level. 
Fig.~\ref{subfig:decision_boundary-A} shows this threshold in the green dotted line.

\textbf{Lines~$10$--$11$}. When $1 - \epsilon_0 < \mu_t< (1 - \frac{1}{2^{|W|}})^2 + \epsilon_1$, we cannot conclude anything with sufficient confidence. This case is illustrated in Fig.~\ref{subfig:decision_boundary-A} between the green and red dotted lines, e.g., the gray triangle. Let us call this region the ``uncertain region''. 
If we can somehow increase the solution ratio (of course, provided that $f$ has solutions) such that its corresponding $\widehat{I_{d,m,K}}$ jumps out of the uncertain region, then our approach can determine a conclusive result. To achieve this goal, we adopt \emph{Shannon expansion}~\cite{Shannon49}, which is briefly introduced in the following. Suppose $V = \{ v_1, v_2, \dots, v_n \}$ is a set of Boolean variables and $f$ is a formula over $V$. We can perform Shannon expansion for $f$ based on a variable $v_i \in V$. We use $f_{[v_i]}$ to denote the Shannon-expanded formula, which is defined as $f_{[v_i]} = f_{v_i \leftarrow 0} \vee f_{v_i \leftarrow 1}$, where $f_{v_i \leftarrow 0}$ denotes the formula obtained by replacing all occurrences of $v_i$ in $f$ with $0$, while $f_{v_i \leftarrow 1}$ denotes the formula obtained by replacing all occurrences of $v_i$ in $f$ with $1$.
We extend the notion of Shannon expansion to a set of Boolean variables. Let $H = \{v_i, v_j, \ldots, v_k\}$ be a subset of $V$. We use the notation $f_{[H]}$ to denote the formula $(((f_{[v_i]})_{[v_j]})_{\ldots})_{[v_k]}$ by applying Shannon expansion to $f$ recursively for each variable in $H$ (the order does not matter). Theorem~\ref{thm:Shannon_Expansion} proves three nice properties about the Shannon-expanded formula $f_{v_i}$.

\begin{restatable}[Shannon Expansion]{thm}{ShannonExpansion} \label{thm:Shannon_Expansion} ~\\[-4mm]
\begin{enumerate}
    \item $f$ has no solutions $\iff$ $f_{[v_i]}$ has no solutions.
    \item $r_f \leq r_{f_{[v_i]}} \leq 2 \cdot r_f$, where $r_f$ and $ r_{f_{[v_i]}}$ are the solution ratios of $f$ and $f_{[v_i]}$.
    \item $f_{[V]}$ is either $f_\top$ with $r=1$, or $f_\bot$ with $r=0$.
\end{enumerate}
\end{restatable}

Theorem~\ref{thm:Shannon_Expansion} suggests that if the solution ratio of formula $f$ is too small such that its corresponding $\mu_t$ falls into the uncertain region, we can pick one variable $v_i \in V$ to construct its corresponding Shannon-expanded formula $f_{[v_i]}$, and then apply $\mathtt{QEC}$ to $f_{[v_i]}$. Since the formula $f_{[v_i]}$ has a higher solution ratio, its $\mu_t$ would be smaller and exit the uncertain region. If it does not, we can apply Shannon expansion again on formula $f_{[v_i]}$ by picking another variable $v_j \in V \setminus \{v_i\}$. This process is performed recursively until we have a conclusive result, i.e., the value of $\mu_t$ is not in the uncertain region. Algorithm~\ref{alg:QSAT} maintains a set of variables $H$ to perform Shannon expansion on formula $f$. Suppose $H = \{x_1, x_2, \ldots, x_h \}$ has $h$ Boolean variables. The general form of $f_{[H]} = \bigvee_{x_i \in \{0,1\}} f_{x_1 x_2 \cdots x_h}$ for $i \in \{1,2,\ldots, h \}$ is the disjunction of $2^h$ CNF sub-formulas.
Each sub-formula $f_{x_1 x_2 \cdots x_h}$ can be obtained by assigning each variable $x_i$ either $0$ or $1$. For example, $f_{0 \ldots 00}$ is obtained by assigning $0$ to each variable, while $f_{0 \ldots 01}$ is obtained by assigning $1$ to $x_h$ and $0$ to all the others, and so on. Fig.~\ref{fig:Quantum_Circuit_Shannon_Expansion} shows how to construct the quantum circuit $\Ddot{U}_{f_{[H]}}$ for $\mathtt{QEC}$. The (pink) circuit for each sub-formula $f_{x_1 x_2 \cdots x_h}$ can be constructed based on the scheme in Fig.~\ref{subfig:Uf_Impl}. The qubit $\ket{F}$ is for the result of $f_{[H]}$. Theorem~\ref{thm:ShanExp_Circuit_Correct} proves that the construction of $\Ddot{U}_{f_{[H]}}$ still preserves the desirable property in Theorem~\ref{thm:TightBound}.

\begin{restatable}[]{thm}{ShanExpCircuitCorrect} \label{thm:ShanExp_Circuit_Correct}
    $\Ddot{U}_{f_{[H]}}(\ket{\vec{v}}\ket{F}) = (-1)^{f_{[H]}(\vec{v}) \wedge F} \ket{\vec{v}} \ket{F}$, where $\vec{v} \in V \setminus H$.
\end{restatable}

Theorem~\ref{thm:Q-SAT-Correctness} proves the correctness of \ourwork. Its analysis is discussed in Section~\ref{sec:Discussion}.

\begin{figure}[tb]
\centering
\includegraphics[width=0.86\linewidth]{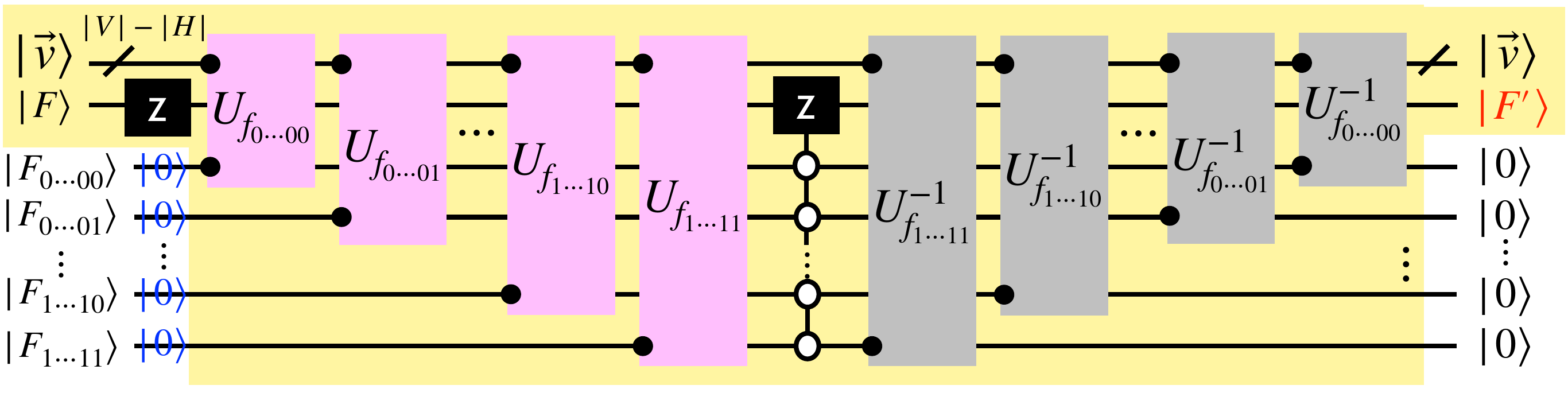}
\caption{Quantum Circuit $\Ddot{U}_{f_{[H]}}$ for Formula $f_{[H]}$}
\label{fig:Quantum_Circuit_Shannon_Expansion}
\end{figure}


\begin{restatable}[]{thm}{QSATCorrectness} \label{thm:Q-SAT-Correctness}
Given a Boolean formula $f$ with $n$ variables and a user-specified confidence level $L \in (0,1)$, Algorithm~\ref{alg:QSAT} returns, in $n$ iterations, the correct answer (i.e., it returns $\mbox{yes}$ if $f$ is satisfiable, and returns $\mbox{no}$ otherwise) with probability at least $L$.
\end{restatable}

\section{Analysis and Discussion} \label{sec:Discussion}

In Algorithm~\ref{alg:QSAT}, we use the theory of \emph{normal approximation to sample mean} (Lemma~\ref{lm:Gaussian_Tail_Probability} in Appendix~\ref{sec:UsfulLemmas}) to approximate $\widehat{I_{d,m,K}}$. However, Lemma~\ref{lm:Gaussian_Tail_Probability} has an assumption that the sample size $s$ has to be sufficiently large. In our case, how large is considered sufficient? Lemma~\ref{lm:Asymptotic_Normality} discharges this assumption and proves the lower bound $\hat{s}$ of the sampling size, which depends on the number of measurements $m$, the number of variables $|V|$, and the confidence level $L$. Remarkably, the values of $m$ and $L$ can be considered constants ($m=2$ and $L=0.95$ are the conventional settings, respectively), and $s$ grows only logarithmically with $|V|$, e.g., with $m = 2$ and $L=0.95$ for $10,000$ variables, only 149 samplings are required. This lower bound $\hat{s}$ gives a suggested value of sampling size in practice.

\begin{restatable}[Asymptotic Normality Requirement]{lm}{AsymptoticNormality} \label{lm:Asymptotic_Normality}
Given the number of measurements $m$, to satisfy the requirement for asymptotic normality, the sample size has a lower bound:
$
\hat{s} \ge \frac{64}{9 \cdot \tan^2(\frac{\pi}{2m})} \left[\Phi^{-1}\left(1 - \frac{\delta}{2}\right)\right]^2$, where $\delta = \frac{1 - L}{|V|+1}$.
\end{restatable}

Now, we want to discuss the $(|V|+1)$-th iteration of Algorithm~\ref{alg:QSAT}, in which we have $W = \emptyset$, $d = 2$, and $H = V$. This is the very last iteration to handle the theoretical worst case where no conclusive result has been made in the previous $|V|$ iterations. In this case, Theorem~\ref{thm:Shannon_Expansion} tells us that $f_{[V]}$ is either $f_\bot$ or $f_\top$. As we want to have a conclusive result in this iteration, we need to make sure the uncertain region is empty, as illustrated in Fig.~\ref{subfig:decision_boundary-B}. It is easy to verify that $\sigma(X_{2,m,0}) = \sigma(X_{2,m,1}) = \frac{1}{2} \tan\left(\frac{\pi}{2m}\right)$ based on Theorems~\ref{thm:StandardVar4NoSol} and~\ref{thm:StandardVar4OneSol}. Thus, given the same number of samples $s$ and error rate $\delta$, the error range $\epsilon_0$ and $\epsilon_1$ would be the same. To make the uncertain region empty, we need the assumption of $\epsilon_0 = \epsilon_1 = \frac{1}{2}$. By solving for $s$ in the equation $2 \left( 1 - \Phi\left( \frac{ \epsilon \sqrt{s}}{\sigma(X_{2,m,0})} \right) \right) = \frac{1 - L}{|V|+1}$ with the condition $\epsilon = \frac{1}{2}$, we can get the number of samplings $s'$ to discharge this assumption. Lemma~\ref{lm:sampling_sufficiency} proves that $\hat{s}$ is greater than $s'$ such that the uncertain region is guaranteed to be empty, i.e., Algorithm~\ref{alg:QSAT} terminates either in Line~$7$ or Line~$8$ in the $(|V|+1)$-th iteration.



\begin{restatable}[]{lm}{SamplingSufficiency} \label{lm:sampling_sufficiency}
The sampling size $\hat{s}$ defined in Lemma~\ref{lm:Asymptotic_Normality} guarantees that the $(|V|+1)$-th iteration of Algorithm~\ref{alg:QSAT} gives a conclusive result for any number of measurement types $m \ge 2$.
\end{restatable}

Let us discuss the complexity of \ourwork. Algorithm~\ref{alg:QSAT} performs $s$ samplings for each iteration and executes for at most $|V|+1$ iterations.
The total number of samplings required in the worst case is $s \cdot (|V|+1)$. The worst case happens when the formula has either no solution ($K=0$) or one solution ($K=1$). This is because the solution ratios of these two cases are too close to be differentiated, and thus Shannon expansion is applied (at most $|V|$ times) to widen the difference between them.

To conclude the worst case time complexity, we need to discuss the execution time for each sampling.
In each iteration of Algorithm~\ref{alg:QSAT}, if needed, Shannon expansion is applied on the current input formula $f$ to get its expanded formula $f_1 \vee f_2$ as the input for the next iteration. That is, every Shannon expansion will double the number of disjunctive sub-formulas in the expanded formula, though each sub-formula itself is simpler than $f$. This makes the number of sub-formulas grow exponentially with the number of Shannon expansions being performed. The implication is that the quantum circuit for the Shannon-expanded formula, as shown in Fig.~\ref{fig:Quantum_Circuit_Shannon_Expansion}, increases in depth. That is, each Shannon expansion doubles the execution time of samplings. Thus, the worst case time complexity of Algorithm~\ref{alg:QSAT} becomes $O(|V| \cdot 2^{|V|})$. Fortunately, the probability of the worst case is extremely small. Let us calculate the probability. The total number of assignments is $N = 2^{|V|}$, and the total number of solution arrangements is $2^N$ because each assignment could be a solution or not. If $f$ has $K$ solutions, there will be $\binom{N}{K}$ possible cases, and the probability of $f$ having $K$ solutions is $\frac{\binom{N}{K}}{2^N}$. The worst case happens when $K = 0$ or $K=1$ with the probability $\frac{\binom{N}{0} + \binom{N}{1}}{2^N} = \frac{1 + N}{2^N}$, which is extremely small when $|V|$ is large. Since the probability of the worst case is extremely small, it is more practical to consider the expected time complexity, analogous to the average-case analysis of Quicksort. 
We analyze the average case under the uniform random Boolean function model, where an 
$n$-variable Boolean function is sampled uniformly from all Boolean functions.
Theorem~\ref{thm:expected_iterations} proves that the \emph{expected} number of Shannon expansions required is $O(1)$, which implies that the \emph{expected} execution time of samplings remains in the same order. Theorem~\ref{thm:expected_exec_time} further proves that the \emph{expected} time complexity of Algorithm~\ref{alg:QSAT} is $s \cdot O(1)$. Our experimental results (c.f.~Section~\ref{sec:Evaluation}) also support our theoretical claims. 

\begin{restatable}[]{thm}{ExpectedIterations} \label{thm:expected_iterations}
The \emph{expected} number of Shannon expansions in \ourwork~is $O(1)$.
\end{restatable}

\vspace{-10px}
\begin{restatable}[]{thm}{ExpectedExecTime} \label{thm:expected_exec_time}
The \emph{expected} execution time complexity of \ourwork~is $s \cdot O(1)$.
\end{restatable}

If this increase in quantum circuit length due to Shannon expansion is a concern, the technique proposed in~\cite{LWC24} can be further adopted to process each sub-formula in parallel, reducing the exponentially growing length to a linear one, which is our future work. Another way to proactively avoid the worst case is to terminate Algorithm~\ref{alg:QSAT} earlier. If Algorithm~\ref{alg:QSAT} is terminated in $t$-th iteration, we could predict the result using the current $\mu_t$ based on its \emph{Bayes factor}~\cite{KR95}. Due to the page limit, we leave this approach in Appendix~\ref{apx:EarlyTermination}.

Finally, we discuss the extraction of the solution assignment. Given a formula $f$, \ourwork~returns the estimated solution ratio $r$ with at least the user-specified confidence level $L$. If a user only cares about whether $f$ has solutions or not, he/she can stop here. If a solution is required,
the estimated solution ratio $r$ fulfills the prerequisite of Grover's algorithm~\cite{Grover96}, i.e., Grover's algorithm can be applied to get the solution to $f$. We will not get into the details of it but summarize the steps as follows:
\begin{enumerate}
    \item Bring each qubit for the $|V|$ variables to the superposition state $\frac{1}{\sqrt{2}}(\ket{0} + \ket{1})$.
    \item Use $\Ddot{U}_f$ in Fig~\ref{subfig:UfDotDot}, as the \emph{oracle} required by Grover's algorithm to mark the solutions.
    \item Use the standard \emph{diffuser}~\cite{FS19,LWC24} to amplify the probability of solutions being measured.
    \item Repeat the Grover iteration (Step~$2$ followed by Step~$3$) for $\lfloor \frac{\pi}{4} \sqrt{\frac{1}{r}} \rfloor$ times.
    \item Measure the input qubits for the $|V|$ variables.
\end{enumerate}

Based on the theory of Grover's algorithm~\cite{Grover96}, if $f$ is satisfiable ($r \neq 0$), a solution will be measured with almost $100\%$ probability, provided that Grover iteration (Step~$4$) is repeated exactly $\lfloor \frac{\pi}{4} \sqrt{\frac{1}{r}} \rfloor$ times; otherwise, the probability will decrease (that is why $r$ is important to Grover's algorithm).
The user can then validate the solution easily by trying the assignment on $f$ to have $100\%$ confidence. If $f$ is unsatisfiable ($r=0$), each assignment could be measured with equal probability, but none of them will satisfy $f$. This constitutes a validation of the result of \ourwork. With the solution ratio $r$ given by \ourwork, the user can perform Grover's algorithm (Steps~$1$--$5$ mentioned above) for a number of rounds, say $10$. During the $10$ rounds, if a solution is measured, the user can stop and conclude $f$ is satisfiable. If no solutions are found during the $10$ rounds, the user can conclude $f$ is unsatisfiable.

\section{Evaluation} \label{sec:Evaluation}

We did not evaluate \ourwork~on real quantum computers because the noise 
on quantum computers is still significant. Instead, we evaluate our approach by simulations on the IBM Qiskit~\cite{GT19} platform. 
Experiments were conducted on a workstation with an Intel Xeon Gold 6226R CPU and 1.1~TiB RAM, running Ubuntu 22.04.3. \ourwork, is implemented in Python using the \texttt{AerSimulator} backend (Qiskit-Aer 0.17.2). 
The source code of \ourwork~and the benchmarking scripts are available at~\cite{SourceCode}.

We evaluate \ourwork~on formulas with $8$ variables. One may think that the number of variables is small, but simulation of qubits takes time. For $8$ Boolean variables, a simulation of $18$ qubits is required (Alice's part: $8$ qubits for variables and $1$ qubit for the formula result; Bob's part is the same). Our experiments took $10$ days to finish. In fact, the number of variables is not the main focus of our evaluation. The focus is to validate our theoretical claims: (a) Given a confidence level $L$ and a formula, \ourwork~is able to return the correct classification with at least $L$ confidence. (b) The expected number of Shannon expansions is indeed $O(1)$.


\begin{figure}[tb]
\begin{subfigure}[b]{0.45\textwidth}
\includegraphics[width=\linewidth]{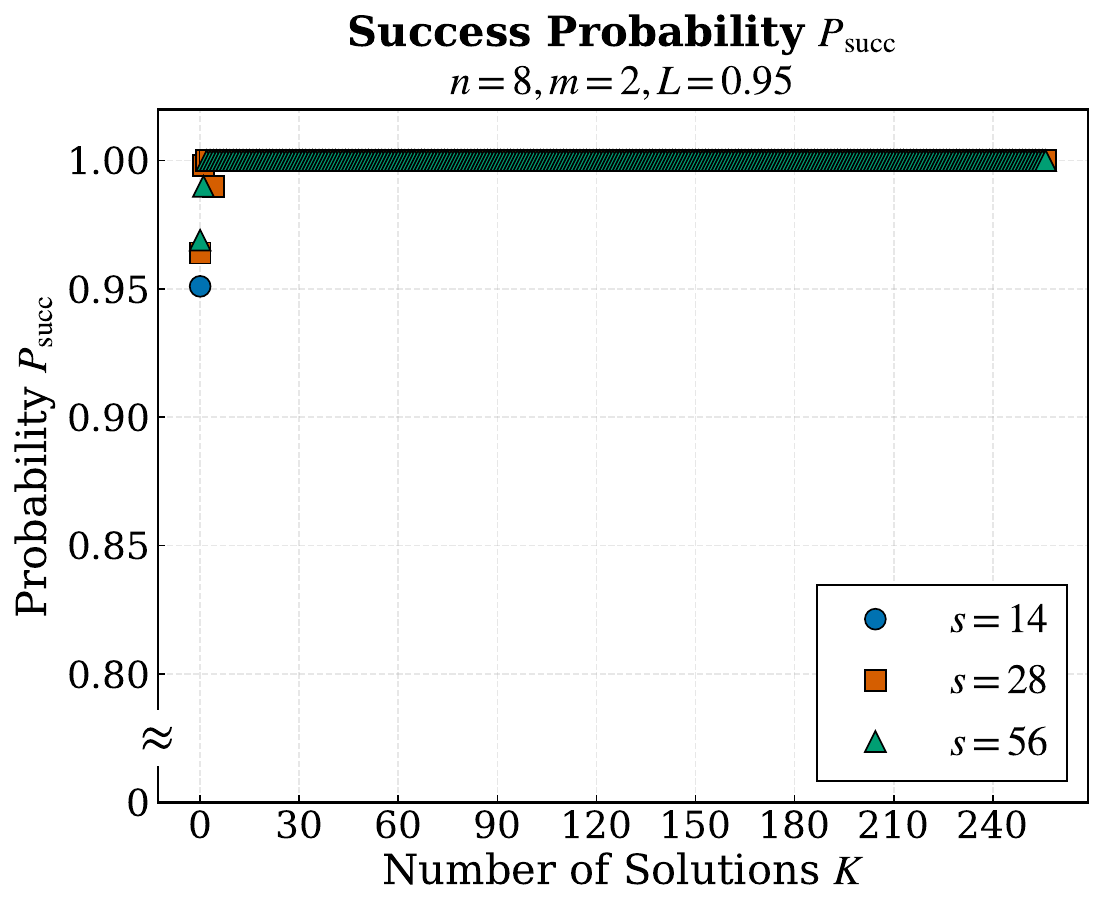}
\caption{Success Probability} \label{subfig:SuccessProbability}
\end{subfigure}
~
\begin{subfigure}[b]{0.49\textwidth}
\includegraphics[width=\linewidth]{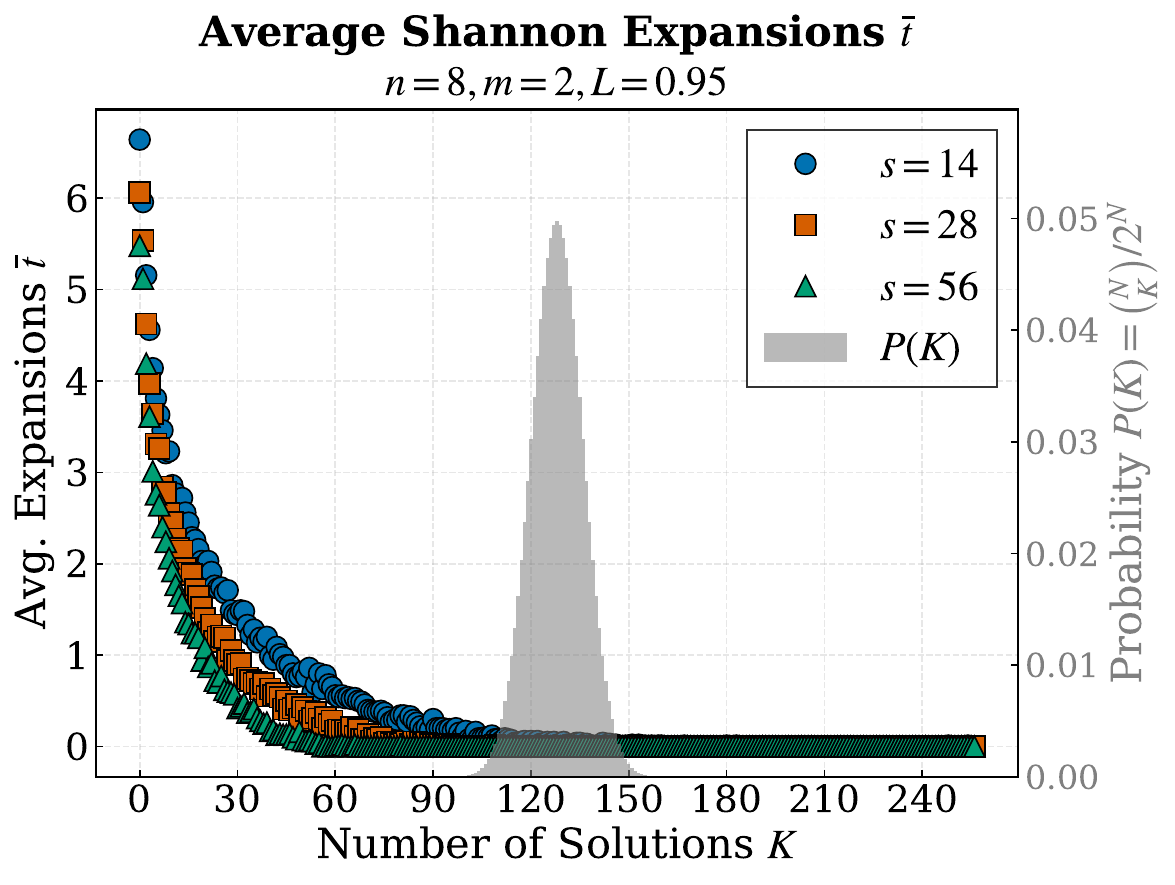}
\caption{Number of Shannon Expansions} \label{subfig:ShannonExpansion}
\end{subfigure}
\caption{Experimental Results of \ourwork}
\label{fig:Effectiveness_Experiments}
\vspace{-10px}
\end{figure}

The experimental results for validating Claim~{(a)} are shown in Fig.~\ref{subfig:SuccessProbability}. We randomly generated $1,000$ formulas (over $8$ variables) with $K$ solutions for each $K \in \{0, 1, \ldots, 256\}$, respectively. To simplify this process, we did not generate the actual Boolean formulas. Instead, we generated their corresponding matrices, which is easier. Based on the proof of Lemma~\ref{lm:FormulaTrace}, for a formula $f$ with $K$ solutions, its corresponding matrix can be obtained by constructing an identity matrix first and then putting a `$-1$' phase in $K$ of the diagonal elements. We use the standard setting for $m=2$ and $L=0.95$. Fig.~\ref{subfig:SuccessProbability} shows the success probability $\mathbb{P}_{\mathtt{suc}}$ for different $K$ and sampling number $s \in \{14, 28, 56\}$, where $\mathbb{P}_{\mathtt{suc}}$ is defined as the number of correct classifications given by \ourwork~out of the $1,000$ formulas.
Notice that $s=56$ meets the lower bound requirement in Lemma~\ref{lm:Asymptotic_Normality}. We did not go beyond $56$ samplings because the sampling number $s$ should be much smaller than $2^8 = 256$; otherwise, it becomes like an exhaustive search. We also consider the cases of $s=14$ and $s=28$ to evaluate the empirical robustness of our approach. 
The smallest sampling size in our experiments is $s=14$, which is greater than the number of sampling size $s'$, discussed in Section~\ref{sec:Discussion} to make sure that Algorithm~\ref{alg:QSAT} always gives a conclusive result. In the case of our experiments, $s' = 8$. That is, the sampling size $s \in \{14,28,56\}$ that we used in our experiments is big enough to make sure that Algorithm~\ref{alg:QSAT} always gives a conclusive result.
One can observe that the worst success probability occurs when $K$ is close to $0$. However, it is still higher than the confidence level $0.95$, which supports our Claim~{(a)}. If $K$ is more than $5$, the success probability is already able to achieve $100\%$. Also, the more samples are performed, the higher success probability there will be. 

The experimental results for validating Claim~{(b)} are shown in Fig.~\ref{subfig:ShannonExpansion}. The same setting applies, and $1,000$ formulas are randomly generated for each $K \in \{ 0,1,\ldots, 256 \}$, respectively. There are two parts shown in Fig.~\ref{subfig:ShannonExpansion}. The regular (left) $y$-axis shows the average number of Shannon expansions performed by \ourwork~among the $1,000$ formulas. The gray (right) $y$-axis shows the probability distributions for all $K$, which is plotted based on $\frac{\binom{N}{K}}{2^N}$, as discussed in Section~\ref{sec:Discussion}, where $N = 2^{8}$. One can observe that the maximum number of Shannon expansions occurs when $K$ is close to $0$, but its occurring probability is also close to $0$. We have also calculated, for $s=14$, the average number of Shannon expansions performed by \ourwork~among the $256$ cases is $0.51$, which supports our Claim~{(b)} that the expected number of Shannon expansions is $O(1)$. One can also find that the more samplings are performed, the fewer Shannon expansions are required. 

\section{Related Works} \label{sec:RelatedWorks}

Grover’s search algorithm~\cite{Grover96} achieves a quadratic speedup for finding a satisfying assignment, provided that at least one exists. In practice, however, the number of satisfying assignments must be known or estimated in advance to determine the optimal number of Grover iterations. Quantum counting~\cite{Mos98,BHMT02,AR21} addresses this requirement by combining amplitude amplification with quantum phase estimation (QPE) to approximate the solution count $K$. These approaches are therefore complementary: Quantum counting estimates $K$, while Grover’s algorithm uses that estimate to search for a solution. However, the reliance on QPE in quantum counting significantly increases the circuit depth and resource overhead. This limitation motivates the search for alternative solvers that can determine satisfiability without requiring a priori knowledge of $K$ or expensive estimation procedures.

Our work provides an alternative approach. Instead of searching for assignments, we statistically determine the satisfiability of the underlying formula, using entanglement-based correlation tests inspired by the quantum equivalence checking (QEC) protocol~\cite{SW22}. This allows unsatisfiable cases to be identified without running Grover’s algorithm, while still enabling an efficient estimation of $K$ when solutions exist.

\section{Conclusion and Future Work} \label{sec:Conclusion}

We present \ourwork, a quantum SAT solver based on entanglement and equivalence checking. Compared to state-of-the-art Grover-based quantum SAT solvers, our approach does not require prior knowledge of the number of solutions to the Boolean formula. We have proved that the \emph{expected} time complexity of \ourwork~is constant under a random Boolean function model. Our experiments also support our theoretical claim. For future work, we plan to adopt the technique in~\cite{LWC24} to improve the worst-case time complexity.



\clearpage


\bibliographystyle{plainurl}
\bibliography{reference}

@MISC{SourceCode,
  title = {\ourwork~Source Code},
  howpublished={\\ \url{https://github.com/n26124939/Quentangle-SAT}}
}

@software{GT19,
author       = {Gadi Aleksandrowicz and
                  Thomas Alexander and
                  Panagiotis Barkoutsos and others},
  title        = {{Qiskit: An Open-source Framework for Quantum 
                   Computing}},
  month        = jan,
  year         = 2019,
  publisher    = {Zenodo},
  version      = {0.7.2},
  doi          = {10.5281/zenodo.2562111},
  url          = {https://doi.org/10.5281/zenodo.2562111}
}

@InProceedings{LWC24,
author="Lin, Shang-Wei
and Wang, Tzu-Fan
and Chen, Yean-Ru
and Hou, Zhe
and San{\'a}n, David
and Teo, Yon Shin",
title="A Parallel and Distributed Quantum SAT Solver Based on Entanglement and Teleportation",
booktitle="Tools and Algorithms for the Construction and Analysis of Systems (TACAS)",
year="2024",
pages="363--382",
isbn="978-3-031-57249-4"
}

@InProceedings{Mos98,
author="Michele Mosca",
title="Quantum Searching, Counting and Amplitude Amplification by Eigenvector Analysis",
booktitle="International Workshop on Randomized Algorithms",
year="1998",
pages="90--100"
}

@book{Hefferon20,
  author = {Hefferon, Jim},
  year = {2020},
  title = {Linear Algebra},
  publisher = {},
  edition = {4}
}

@article{ADR82,
author = {Aspect, Alain and Dalibard, Jean and Roger, Gérard},
title = {Experimental Test of Bell's Inequalities Using Time-Varying Analyzers},
year = {1982},
volume = {1804},
number = {49},
journal = {Physical Review Letters},
pages = {},
}

@article{CHSH69,
author = {Clauser, John F. and Horne, Michael A. and Shimony, Abner and Holt, Richard A.},
title = {Proposed Experiment to Test Local Hidden-Variable Theories},
year = {1969},
volume = {880},
number = {23},
journal = {Physical Review Letters},
pages = {},
}

@article{Bell64,
author = {John Stewart Bell},
title = {On the {E}instein Podolsky Rosen Paradox},
year = {1964},
volume = {195},
number = {1},
journal = {Physics Physique Fizika},
pages = {},
}

@article{SW22,
author = {Sun, Weixiao and Wei, Zhaohui},
title = {Equivalence Checking of Quantum Circuits by Nonlocality},
year = {2022},
volume = {8},
number = {139},
doi = {https://doi.org/10.1038/s41534-022-00652-x},
journal = {npj Quantum Information},
pages = {},
}

@article{SATW17,
author = {Salavrakos, Alexia and Augusiak, Remigiusz and Tura, Jordi and Wittek, Peter and  Ac\'{i}n, Antonio and  Pironio, Stefano},
title = {Bell Inequalities Tailored to Maximally Entangled States},
year = {2017},
volume = {119},
number = {040402},
doi = {https://doi.org/10.1103/PhysRevLett.119.040402},
journal = {Physical Review Letters},
pages = {},
}

@article{LCWC23,
author = {Lin, Shang-Wei and Chen, Si-Han and Wang, Tzu-Fan and Chen, Yean-Ru},
title = {A Quantum {SMT} Solver for Bit-Vector Theory},
year = {2023},
volume = {arXiv:2303.09353},
number = {},
doi = {arXiv:2303.09353},
journal = {arXiv preprint},
pages = {},
url = {https://doi.org/10.48550/arXiv.2303.09353}
}

@article{FS19,
author = {Fernandes, Diogo and Silva, Carla and Dutra, In\^{e}s},
title = {Using Grover's Search Quantum Algorithm to Solve Boolean Satisfiability Problems, Part 2},
year = {2019},
issue_date = {Winter 2019},
publisher = {Association for Computing Machinery},
address = {New York, NY, USA},
volume = {26},
number = {2},
issn = {1528-4972},
url = {https://doi.org/10.1145/3368085},
doi = {10.1145/3368085},
journal = {XRDS},
month = {nov},
pages = {68–71},
numpages = {4}
}

@inproceedings{Grover96,
  title={A fast quantum mechanical algorithm for database search},
  author={Grover, Lov K},
  booktitle={Proceedings of the 28th Annual ACM Symposium on Theory of Computing (STOC)},
  pages={212--219},
  year={1996}
}

@article{BHMT02,
  title={Quantum amplitude amplification and estimation},
  author={Brassard, Gilles and Hoyer, Peter and Mosca, Michele and Tapp, Alain},
  journal={arXiv preprint quant-ph/0005055},
  year={2000}
}

@inproceedings{AR21,
  title={Quantum approximate counting, simplified},
  author={Aaronson, Scott and Rall, Patrick},
  booktitle={Symposium on simplicity in algorithms},
  pages={24--32},
  year={2020},
  organization={SIAM}
}

@book{GR14,
  author    = {Gradshteyn, Izrail Solomonovich and Ryzhik, Iosif Moiseevich},
  title     = {Table of Integrals, Series, and Products},
  year      = {2014},
  publisher = {Academic Press},
  edition   = {8th},
  address   = {Burlington, MA}
}

@article{Shannon49,
  title={The synthesis of two-terminal switching circuits},
  author={Shannon, Claude E},
  journal={Bell System Technical Journal},
  volume={28},
  number={1},
  pages={59--98},
  year={1949}
}

@article{Hoeffding63,
  title={Probability inequalities for sums of bounded random variables},
  author={Hoeffding, Wassily},
  journal={Journal of the American Statistical Association},
  volume={58},
  number={301},
  pages={13--30},
  year={1963},
  publisher={Taylor \& Francis}
}

@article{Bonferroni36,
  title={Teoria statistica delle classi e calcolo delle probabilit{\`a}},
  author={Bonferroni, Carlo E},
  journal={Pubblicazioni del R Istituto Superiore di Scienze Economiche e Commerciali di Firenze},
  volume={8},
  pages={3--62},
  year={1936}
}

@book{Wasserman04,
  title={All of Statistics: A Concise Course in Statistical Inference},
  author={Wasserman, Larry},
  year={2004},
  publisher={Springer},
  address={New York},
  note={Chapter 5, Theorem 5.8}
}

@book{BLM13,
    author = {Boucheron, St{\'e}phane and Lugosi, G{\'a}bor and Massart, Pascal},
    title = {Concentration Inequalities: A Nonasymptotic Theory of Independence},
    publisher = {Oxford University Press},
    year = {2013},
    month = {02},
    isbn = {9780199535255},
    doi = {10.1093/acprof:oso/9780199535255.001.0001},
    url = {https://doi.org/10.1093/acprof:oso/9780199535255.001.0001}
}

@article{KR95,
  title={Bayes factors},
  author={Kass, Robert E and Raftery, Adrian E},
  journal={Journal of the American statistical association},
  volume={90},
  number={430},
  pages={773--795},
  year={1995},
  publisher={Taylor \& Francis}
}

@article{Benjamini02,
    title={John W. Tukey's contributions to multiple comparisons},
    author={Benjamini, Yoav and Braun, Henry},
    journal={Annals of Statistics},
    pages={1576--1594},
    year={2002},
    publisher={JSTOR}
}

\clearpage
\appendix
\section*{APPENDIX} \label{sec:Appendix}



\section{Lemmas and Their Proofs} \label{sec:UsfulLemmas}

\FormulaTrace*

\begin{proof}
First of all, let us consider the formula $f$ with no solution. By Theorem~\ref{thm:FalseIdentity}, $\Ddot{U}_f$ is the identity matrix and $\mathsf{Tr}(\Ddot{U}_f) = 2^{n+1}$. Each `$1$' in the diagonal position indicates that $\Ddot{U}_f(\ket{\vec{v}}\ket{F}) = \ket{\vec{v}} \ket{F}$ for $\vec{v} \in \{ 0,1\}^n $ and $F \in \{0,1\}$.
Suppose that we want to add one solution $\vec{v} \in \{0,1\}^n$ to formula $f$, which changes $f(\vec{v})$ from $0$ to $1$ now. By Theorem~\ref{thm:FalseIdentity}, we have $\Ddot{U}_f(\ket{\vec{v}}\ket{1}) = (-1)^{f(\vec{v}) \wedge 1} \ket{\vec{v}} \ket{1} = - \ket{\vec{v}} \ket{1}$. That is to say, the `$1$' in the diagonal position corresponding to $\ket{\vec{v}} \ket{1}$ changes from $1$ to $-1$, which decreases $\mathsf{Tr}(\Ddot{U}_f)$ by $2$. Thus, if we add $K$ solutions to $f$, $\mathsf{Tr}(\Ddot{U}_f) = 2^{n+1} - 2K$.
\end{proof}

\IDMRandomVar*

\begin{proof}
We have $\mathbb{P}\left(\widetilde{(A^1_i)_{\ket{\phi}}} = \widetilde{(B^1_i)_{\ket{\phi}}} + k \right) = \mathbb{P}\left( \widetilde{(B^1_i)_{\ket{\phi}}} = \widetilde{(A^1_{i+1})_{\ket{\phi}}} + k \right)$ for $i \in \{1,2,\ldots, m\}$ based on Eq.~{A$6$} in \cite{SATW17}. If we recall the definition of $\widehat{I_{d,m,K}}(\ket{\phi})$ and rewrite it using the above condition, we have
\[
\begin{split}
\widehat{I_{d,m,K}}(\ket{\phi}) &= \frac{1}{m} \sum_{k=0}^{d-1} \sum_{i=1}^{m} \alpha_k 
    \left[
        \mathbb{P}\left( \widetilde{(A^1_i)_{\ket{\phi}}} = \widetilde{(B^1_i)_{\ket{\phi}}} + k \right) + \mathbb{P} \left( \widetilde{(B^1_i)_{\ket{\phi}}} = \widetilde{(A^1_{i+1})_{\ket{\phi}}} + k \right)
    \right] \\
&= \sum_{k=0}^{d-1} \sum_{i=1}^{m} \frac{\alpha_k }{m} 
    \left[
        2 \cdot \mathbb{P}\left( \widetilde{(A^1_i)_{\ket{\phi}}} = \widetilde{(B^1_i)_{\ket{\phi}}} + k \right)
    \right] \\
&= \sum_{k=0}^{d-1} 2 \alpha_k 
\left[
\frac{1 }{m}
\sum_{i=1}^{m}      
        \mathbb{P}\left( \widetilde{(A^1_i)_{\ket{\phi}}} = \widetilde{(B^1_i)_{\ket{\phi}}} + k \right)
    \right] = \sum_{k=0}^{d-1} x_k \cdot \mathbb{P}(x_k) = \mathbb{E}(X_{d,m,K})
\end{split}
\]
\end{proof}

\AsymptoticNormality*

\begin{proof}
According to the Bernstein inequality~\cite[Corollary~2.11]{BLM13}, for a sum $S = \sum_{i=1}^s (X_i - \mu)$ where $X_i - \mu \le b$, the tail probability is bounded by:
\begin{equation}
    \mathbb{P}(S \ge t) \le \exp\left( - \frac{t^2}{2(v + ct)} \right),
\end{equation}
where $v = s\sigma^2$ and $c = b/3$. By setting $t = s\epsilon$ and considering the two-tailed case, we obtain:
\begin{equation}
    \mathbb{P}(|\overline{X} - \mu| \ge \epsilon) \le 2 \exp\left( - \frac{s\epsilon^2}{2\sigma^2 + \frac{2}{3}b\epsilon} \right).
\end{equation}

In Algorithm~\ref{alg:QEC}, each measurement outcome $X_t$ is defined as $2\alpha_k$. Since $|\alpha_k| \le 1$ as established in~\cite{SW22}, we have $X_{d,m,K} \in [-2, 2]$. To strictly bound the deviation $|X_t - \mu|$ for any potential mean $\mu \in [-2, 2]$, we set $b=4$. The transition to the variance-dominated Gaussian regime occurs when $2\sigma^2$ dominates the correction term $\frac{2}{3}b\epsilon$, leading to the balance criterion $\sigma^2 \approx \frac{4}{3}\epsilon$.

Given that $\epsilon = \frac{\sigma}{\sqrt{s}} \Phi^{-1}(1 - \delta/2)$ for a confidence level $1-\delta$, substituting $\epsilon$ into the criterion yields the required sample size for Gaussian convergence:
\begin{equation}
    \sigma^2 \approx \frac{4}{3} \cdot \frac{\sigma}{\sqrt{s}} \Phi^{-1}\left(1 - \frac{\delta}{2}\right) \implies s \ge \frac{16}{9 \sigma^2} \left[\Phi^{-1}\left(1 - \frac{\delta}{2}\right)\right]^2. \label{eq:s_sigma_inverse_relation}
\end{equation}

As Eq.~\eqref{eq:s_sigma_inverse_relation} indicates, $s$ is inversely proportional to the variance $\sigma^2$. This implies that cases with smaller variances require more samples to reach the Gaussian regime. We therefore compare the variances for the UNSAT ($K=0$) and SAT ($K=1$) cases. According to Theorem~\ref{thm:StandardVar4NoSol} and Theorem~\ref{thm:StandardVar4OneSol}:
\[
    \sigma^2(X_{d,m,1}) = \sigma^2(X_{d,m,0}) + \frac{4(d-1)(d-2)^2}{d^4}.
\]
Because $4(d-1)(d-2)^2/d^4 \ge 0$ for all $d \ge 2$, it follows that $\sigma^2(X_{d,m,1}) \ge \sigma^2(X_{d,m,0})$. Furthermore, the variance for the $K=0$ case is given by:
\[
    \sigma^2(X_{d,m,0}) = \frac{1}{3}\left( 1 - \frac{1}{d^2} \right) \tan^2\left(\frac{\pi}{2m}\right).
\]
Since $\sigma^2(X_{d,m,0})$ is strictly increasing with respect to $d$, the global minimum variance occurs at $d=2$. Substituting $\sigma^2(X_{2,m,0})$ into Eq.~\eqref{eq:s_sigma_inverse_relation} concludes the proof.
\end{proof}

\SamplingSufficiency*

\begin{proof}
Let $\delta = \frac{1-L}{|V|+1}$. In the $(|V|+1)$-th iteration of Algorithm~\ref{alg:QSAT}, we have $W = \emptyset$, $d = 2$, and $H = V$. In this case, Theorem~\ref{thm:Shannon_Expansion} tells us that $f_{[V]}$ is either $f_\bot$ or $f_\top$. As we want to have a conclusive result in this iteration, we need to make sure the uncertain region is empty, as illustrated in Fig.~\ref{subfig:decision_boundary-B}. It is easy to verify that $\sigma(X_{2,m,0}) = \sigma(X_{2,m,1}) = \frac{1}{2} \tan\left(\frac{\pi}{2m}\right)$ based on Theorems~\ref{thm:StandardVar4NoSol} and~\ref{thm:StandardVar4OneSol}. Thus, given the same number of samples $s$ and error rate $\delta$, the error range $\epsilon_0$ and $\epsilon_1$ would be the same. To make the uncertain region empty, it is required that $\epsilon_0 = \epsilon_1 = \frac{1}{2}$. 
By solving the following equation, we can obtain the lower bound of the sampling size $s'$ to make the uncertain region empty:
\begin{align}
    2 \left( 1 - \Phi\left( \frac{\frac{1}{2} \sqrt{s'}}{\sigma(X_{2,m,0})} \right) \right) = \delta &\implies 1 - \Phi\left( \frac{\sqrt{s'}}{2\sigma(X_{2,m,0})} \right) = \frac{\delta}{2} \nonumber \\
    &\implies \frac{\sqrt{s'}}{2\sigma(X_{2,m,0})} = \Phi^{-1}\left( 1 - \frac{\delta}{2} \right) \nonumber \\
    &\implies s' = 4 \sigma^2(X_{2,m,0}) \left[ \Phi^{-1}\left( 1 - \frac{\delta}{2} \right) \right]^2.
\end{align}
In contrast, the threshold $\hat{s}$ for asymptotic normality from Lemma~\ref{lm:Asymptotic_Normality} is:
\begin{equation}
    \hat{s} \ge \frac{16}{9 \sigma^2(X_{2,m,0})} \left[ \Phi^{-1}\left( 1 - \frac{\delta}{2} \right) \right]^2.
\end{equation}
The ratio $R$ between the theoretical requirement and the decision requirement is:
\begin{equation}
    R = \frac{\hat{s}}{s'} = \frac{ \frac{16}{9 \sigma^2(X_{2,m,0})} \left[ \Phi^{-1}\left( 1 - \frac{\delta}{2} \right) \right]^2 }{ 4 \sigma^2(X_{2,m,0}) \left[ \Phi^{-1}\left( 1 - \frac{\delta}{2} \right) \right]^2 } = \frac{4}{9 \sigma^4(X_{2,m,0})}.
\end{equation}
By substituting $d=2$ into the standard deviation formula from Theorem~\ref{thm:StandardVar4NoSol}, we obtain:
\begin{equation}
    \sigma(X_{2,m,0}) = \frac{1}{2} \sqrt{\frac{2^2 - 1}{3}} \tan\left( \frac{\pi}{2m} \right) = \frac{1}{2} \tan\left( \frac{\pi}{2m} \right) \implies \sigma^2(X_{2,m,0}) = \frac{1}{4} \tan^2\left( \frac{\pi}{2m} \right).
\end{equation}
The condition $\hat{s} \ge s'$ holds if $R \ge 1$, which requires $\sigma^2(X_{2,m,0}) \le 2/3 \approx 0.667$. Since $\tan(\frac{\pi}{2m})$ strictly decreases as $m$ increases, the variance is maximized at $m=2$:
\begin{equation}
    \sigma^2(X_{2,2,0}) = \frac{1}{4} \tan^2\left( \frac{\pi}{4} \right) = 0.25.
\end{equation}
Given that the worst-case variance $0.25$ is less than $0.667$, the ratio $R \approx 7.11 > 1$ is maintained for all $m \ge 2$. Thus, the sample size required for theoretical normality always ensures a sufficient sample size for a definitive SAT/UNSAT classification.
\end{proof}

\begin{restatable}[Maximal Variance Configuration for $K=1$]{lm}{MaxVarianceKOne} \label{lemma:max_variance_K1} ~\\
For the random variable $X_{d,m,1}$, the variance is maximized when the solution index $q^*$ is located at the boundary of the search space ($q^*=0$ or $q^*=d-1$).
\end{restatable}

\begin{proof}
Since the expected value $\mathbb{E}[X_{d,m,1}]$ is constant for a fixed number of solutions, maximizing the variance is equivalent to maximizing the second moment $\mathbb{E}[X_{d,m,1}^2]$.
From Lemma~\ref{lm:general_second_moment}, the second moment is a weighted sum involving $(2\alpha_k)^2$.
Due to the singularity of $\cot(x)$ as $x \to 0$, the squared coefficient $\alpha_0^2$ dominates the sum (scaling as $O(1)$ versus $O(m^{-2})$ for $k \ge 1$). Thus, the contribution from $k \neq 0$ is negligible, and it suffices to maximize $\mathbb{P}(x_0)$.

Consider a single solution at index $q^*$. We set $\gamma_{q^*} = -1/\sqrt{d}$ and $\gamma_q = 1/\sqrt{d}$ for $q \neq q^*$. Let $r = \omega^{\frac{1}{2m}}$, the probability term $\mathbb{P}(x_0)$ is expressed as:
\begin{equation}
    \mathbb{P}(x_0) = d \left| \frac{1}{d} \left( \sum_{q=0}^{d-1} \frac{1}{\sqrt{d}} r^q - \frac{2}{\sqrt{d}} r^{q^*} \right) \right|^2 = \frac{1}{d^2} \left| \sum_{q=0}^{d-1} r^q - 2r^{q^*} \right|^2.
\end{equation}

We explicitly evaluate the phases of these two terms. For the summation term:
\begin{equation}
\begin{split}
    \sum_{q=0}^{d-1} r^q &= \frac{1 - r^d}{1 - r} = \frac{1 - (\omega^{\frac{1}{2m}})^d}{1 - \omega^{\frac{1}{2m}}} = \frac{1 - (e^{i \frac{2\pi}{d}})^{\frac{d}{2m}}}{1 - (e^{i \frac{2\pi}{d}})^{\frac{1}{2m}}} = \frac{1 - e^{i \frac{\pi}{m}}}{1 - e^{i \frac{\pi}{dm}}} \\
    &= \frac{e^{i \frac{\pi}{2m}} \left( e^{-i \frac{\pi}{2m}} - e^{i \frac{\pi}{2m}} \right)}{e^{i \frac{\pi}{2dm}} \left( e^{-i \frac{\pi}{2dm}} - e^{i \frac{\pi}{2dm}} \right)} = \frac{e^{i \frac{\pi}{2m}} \left( -2i \sin(\frac{\pi}{2m}) \right)}{e^{i \frac{\pi}{2dm}} \left( -2i \sin(\frac{\pi}{2dm}) \right)} \\
    &= e^{i \left( \frac{\pi}{2m} - \frac{\pi}{2dm} \right)} \cdot \frac{\sin(\frac{\pi}{2m})}{\sin(\frac{\pi}{2dm})} = e^{i \frac{\pi(d-1)}{2dm}} \cdot \frac{\sin(\frac{\pi}{2m})}{\sin(\frac{\pi}{2dm})}.
\end{split}
\end{equation}

For the solution term:
\begin{equation}
\begin{split}
    2r^{q^*} &= 2 \left( \omega^{\frac{1}{2m}} \right)^{q^*} = 2 \left( \left( e^{i \frac{2\pi}{d}} \right)^{\frac{1}{2m}} \right)^{q^*} \\
    &= 2 \left( e^{i \frac{2\pi}{d} \cdot \frac{1}{2m}} \right)^{q^*}  = 2 \left( e^{i \frac{\pi}{dm}} \right)^{q^*} \\
    &= 2 e^{i \frac{\pi q^*}{dm}}.
\end{split}
\end{equation}

The probability $\mathbb{P}(x_0)$ is maximized when the magnitude of the difference between these two vectors is maximized. Geometrically, this occurs when the phase difference $\Delta \phi$ is largest:
\begin{equation}
    \Delta \phi = \left| \frac{\pi(d-1)}{2dm} - \frac{\pi q^*}{dm} \right| = \frac{\pi}{dm} \left| \frac{d-1}{2} - q^* \right|.
\end{equation}
This difference is maximal when the distance $| \frac{d-1}{2} - q^* |$ is largest, i.e., when $q^*$ is at the boundaries $0$ or $d-1$.
\end{proof}

\begin{restatable}[Hoeﬀding’s Inequality~\cite{Hoeffding63}]{lm}{HoeﬀdingInequality} \label{lm:Hoeﬀding_Inequality}
Suppose that $X_1, X_2, \ldots, X_n$ are $n$ independent identical random variables such that $a \leq X_i \leq b$ and $\mathbb{E}(X_i) = \mu$ for $i \in \{1, 2, \ldots, n\}$.  Let $\overline{X} = n^{-1} \sum_i^{n} X_i$ be the sample mean of $X_i$. For any $\epsilon > 0$, we have
$\mathbb{P}(\mid \overline{X} - \mu \mid \geq \epsilon) \, \leq \, \delta$ , where $\delta = 2 \exp\left(- \frac{2 n\epsilon^2}{(b-a)^2}\right)$.
\end{restatable}

\begin{restatable}[Normal Approximation to Sample Mean]{lm}{GaussianTailProbability} 
\label{lm:Gaussian_Tail_Probability}
Suppose that $X_1, X_2, \ldots, X_n$ are $n$ independent and identically distributed random variables with expectation $\mathbb{E}[X_i] = \mu$ and finite variance $\mathtt{Var}(X_i) = \sigma^2$. Let $\overline{X} = n^{-1} \sum_{i=1}^{n} X_i$ be the sample mean. If the same size $n$ is sufficiently large, then for any $\epsilon > 0$, we have
\begin{equation}
\mathbb{P}\left( | \overline{X} - \mu | \ge \epsilon \right) \approx \delta, \,\, \mbox{ where } \,\, \delta = 2 \left( 1 - \Phi\left( \frac{\epsilon \sqrt{n}}{\sigma} \right) \right). \nonumber
\end{equation}
Here, $\Phi$ is the cumulative distribution function (CDF) of the standard normal distribution, defined as: $\Phi(x) = \frac{1}{\sqrt{2\pi}} \int_{-\infty}^{x} (e^{-t^2/2}) dt$.
\end{restatable}

\begin{proof}
By the Central Limit Theorem (Theorem 5.8 in \cite{Wasserman04}), the standardized sample mean $Z_n$ converges in distribution to a standard normal distribution:
\begin{equation}
Z_n = \frac{\overline{X} - \mu}{\sigma / \sqrt{n}} \rightsquigarrow N(0, 1) \nonumber
\end{equation}
To calculate the tail probability $\mathbb{P}(\mid \overline{X} - \mu \mid \ge \epsilon)$, we normalize the expression:
\begin{align}
\mathbb{P}\left( \mid \overline{X} - \mu \mid \ge \epsilon \right) &= \mathbb{P}\left( \frac{\mid \overline{X} - \mu \mid}{\sigma / \sqrt{n}} \ge \frac{\epsilon}{\sigma / \sqrt{n}} \right) = \mathbb{P}\left( \mid Z_n \mid \ge \frac{\epsilon \sqrt{n}}{\sigma} \right) \nonumber
\end{align}
For large $n$, using the symmetry of the standard normal distribution $N(0, 1)$:
\begin{align}
\mathbb{P}\left( \mid Z_n \mid \ge \frac{\epsilon \sqrt{n}}{\sigma} \right) &\approx 2 \cdot \mathbb{P}\left( Z \ge \frac{\epsilon \sqrt{n}}{\sigma} \right) = 2 \left[ 1 - \mathbb{P}\left( Z \le \frac{\epsilon \sqrt{n}}{\sigma} \right) \right] = 2 \left[ 1 - \Phi\left( \frac{\epsilon \sqrt{n}}{\sigma} \right) \right] \nonumber
\end{align}
\end{proof}

\begin{restatable}[]{lm}{IDMM} \label{lm:IDM-M}
Let $\mathcal{S}$ be the subspace generated by the basis $\{ \ket{j}\ket{j} \}_d$, where $d = 2^n$ and $j \in \{0, 1, \ldots, 2^n  -1 \}$. Let $\ket{\phi} \in \mathcal{S}$ be any state orthogonal to $\ket{\Phi_d}$. We have 
$
I_{d,m}(\ket{\phi}) = -m.
$
\end{restatable}
\begin{proof}
Let us recall that $A_x^l = \sum_{a=0}^{d-1} \omega^{ a \cdot l} |a\rangle_{xx}\langle a|$ and $\overline{B}_x^l = (A_x^l)^*$, where 
\[ 
\ket{a}_x = \frac{1}{\sqrt{d}} \sum_{k=0}^{d-1} \omega^{ k \cdot(a - \alpha_x)} |k\rangle.  
\]
Then, we can rewrite $A_x^l \otimes \overline{B}_x^l$ as follows:
\[
A_x^l \otimes \overline{B}_x^l = \sum_{a = 0}^{d-1} \sum_{b=0}^{d-1} \omega^{l (a-b)} \left( \ket{a}_{xx} \! \bra{a} \right) \otimes \left( \ket{b}_{xx} \! \bra{b}^*\right).
\]
And, $\ket{b}_x^*$ can be obtained as follows:
\[ 
\ket{b}_x^* = \frac{1}{\sqrt{d}} \sum_{k=0}^{d-1} \omega^{ -k \cdot(b - \alpha_x)} |k\rangle.  
\]
For the measurement round $x$, let us consider an arbitrary state $\ket{\phi}$ represented by the basis $\{ \ket{a}_x \ket{b}_x^* \}_d$. That is, $\ket{\phi} = \sum_{a=0}^{d-1} \sum_{b=0}^{d-1} h_{ab,x} \ket{a}_x \ket{b}_x^*$ , where $h_{ab,x} \in \mathbb{C}$ and $\sum_{a=0}^{d-1} \sum_{b=0}^{d-1} |h_{ab,x}|^2 = 1$. We have the followings.
\[
\bra{\phi} A_x^l \otimes \overline{B}_x^l \ket{\phi} = \sum_{a=0}^{d-1} \sum_{b=0}^{d-1} |h_{ab,x}|^2 \omega^{l (a-b)}
\]
\[
\begin{split}
\sum_{l=1}^{d-1} \bra{\phi} A_x^l \otimes \overline{B}_x^l \ket{\phi} &= \sum_{l=1}^{d-1} \sum_{a=0}^{d-1} \sum_{b=0}^{d-1} |h_{ab,x}|^2 \omega^{l (a-b)}
 = \sum_{a=0}^{d-1} |h_{aa,x}|^2 (d-1) + \sum_{a \neq b} |h_{ab,x}|^2 (-1) \\
 &= d \sum_{a=0}^{d-1} |h_{aa,x}|^2 - \left( \sum_{a=0}^{d-1} |h_{aa,x}|^2 + \sum_{a \neq b} |h_{ab,x}|^2 \right)
 = d \left( \sum_{a=0}^{d-1} |h_{aa,x}|^2 \right) - 1
\end{split}
\]
Therefore, we can rewrite $I_{d,m}(\ket{\phi})$, as follows:
\[
I_{d,m}(\ket{\phi}) = \sum_{x=1}^{m} \sum_{l=1}^{d-1} \bra{\phi} A_x^l \otimes \overline{B}_x^l \ket{\phi} = \sum_{x=1}^{m} \left( d \sum_{a=0}^{d-1} |h_{aa,x}|^2 - 1 \right) = d \left( \sum_{x=1}^{m} \sum_{a=0}^{d-1} |h_{aa,x}|^2 \right) - m
\]
Now, let us rewrite $\ket{\phi}$ on the basis $\{ \ket{j} \ket{j} \}_d$. That is, $\ket{\phi} = \sum_{k=0}^{d-1} \sum_{j=0}^{d-1} \gamma_{kj} \ket{k} \ket{j}$. Since $\ket{\phi} \in \mathcal{S}$, we have Eq.~\ref{eq:subspace}.
\begin{equation} \label{eq:subspace}
\gamma_{kj} = 0, \mbox{ for all } k,j \in \{0, 1, \ldots, d - 1\} \mbox{ and } k \neq j
\end{equation}
Furthermore, $\ket{\phi}$ is orthogonal to $\ket{\Phi_d}$, we have Eq.~\ref{eq:orthogonal}.
\begin{equation} \label{eq:orthogonal}
\sum_{k=0}^{d-1} \gamma_{kk} = 0
\end{equation}

Now, let us try to figure out the relation between $h_{aa,x}$ and $\gamma_{kj}$. Since
\[
\ket{a}_x \ket{a}_x^* = \frac{1}{d} \sum_{k=0}^{d-1} \sum_{j=0}^{d-1} \omega^{ (k-j) \cdot(a - \alpha_x)} \ket{k} \ket{j},
\]
we can obtain their relation, as follows.
\[
h_{aa,x} = \frac{1}{d} \sum_{k=0}^{d-1} \sum_{j=0}^{d-1} \omega^{ (k-j) \cdot(a - \alpha_x)} \cdot \gamma_{kj}
\]
Because $\ket{\phi} \in \mathcal{S}$ is orthogonal to $\ket{\Phi_d}$, we can use Eq.~\ref{eq:subspace} and Eq.~\ref{eq:orthogonal} to conclude
\[
h_{aa,x} = \frac{1}{d} \sum_{k=0}^{d-1} \omega^0 \cdot \gamma_{kk} = \frac{1}{d} \left(  \sum_{k=0}^{d-1} \gamma_{kk} \right) = 0.
\]
Thus, we prove that
\[
I_{d,m}(\ket{\phi}) = d \left( \sum_{x=1}^{m} \sum_{a=0}^{d-1} |h_{aa,x}|^2 \right) - m = d \times 0 - m =  -m.
\]
\end{proof}

\begin{restatable}[Generalized Second Moment]{lm}{GeneralSecondMoment} \label{lm:general_second_moment} ~\\
Let $\ket{\phi} = \sum_{q=0}^{d-1} \gamma_q \ket{q}\ket{q}$ be a quantum state. Consider the random variable $X_{d,m,K}$ whose expected value is $\widehat{I_{d,m,K}}(\ket{\phi})$. The second moment is:
\begin{equation} \label{eq:general_E_X2}
    \mathbb{E}(X_{d,m,K}^2) = \sum_{k=0}^{d-1} (2\alpha_k)^2 \cdot \mathbb{P}(x_k),
\end{equation}
with the probability $\mathbb{P}(x_k)$ defined as follows, where $\omega = e^{\frac{2\pi \mathbf{i}}{d}}$
\begin{equation}
    \mathbb{P}(x_k) = d \left| \frac{1}{d} \sum_{q = 0}^{d-1} \gamma_q \cdot \omega^{q (k + \frac{1}{2m})} \right|^2.
\end{equation}
\end{restatable}

\begin{proof}
From Lemma~\ref{lm:IDM_RandomVar}, we have established that the random variable $X_{d,m,K}$ takes the values $x_k = 2\alpha_k$ with probability $\mathbb{P}(x_k)$. By the definition of the second moment, it follows directly that:
\begin{equation}
    \mathbb{E}(X_{d,m,K}^2) = \sum_{k=0}^{d-1} x_k^2 \cdot \mathbb{P}(x_k) = \sum_{k=0}^{d-1} (2 \alpha_k)^2 \cdot \mathbb{P}(x_k).
\end{equation}

We derive the explicit form of $\mathbb{P}(x_k)$. With $\omega = e^{2\pi \mathbf{i} / d}$ and the measurement settings $\theta_x = \frac{x-1/2}{m}, \zeta_y = \frac{y}{m}$, the outcome probabilities follow:
\begin{align}
    \mathbb{P}\left(\widetilde{(A^1_x)_{\ket{\phi}}} = a, \widetilde{(B^1_y)_{\ket{\phi}}} = b \right) &= \left| \frac{1}{d} \sum_{q=0}^{d-1} \gamma_q \omega^{q (a-b-\theta_x+\zeta_y)} \right|^2 \tag{by Eq.~(S.4) in \cite{SATW17}} \\
    \mathbb{P}\left(\widetilde{(A^1_x)_{\ket{\phi}}} = \widetilde{(B^1_y)_{\ket{\phi}}} + k \right) &= d \cdot \mathbb{P}\left(\widetilde{(A^1_x)_{\ket{\phi}}} = k, \widetilde{(B^1_y)_{\ket{\phi}}} = 0 \right) \tag{by Eq.~(S.5) in \cite{SATW17}} \nonumber \\
    &= d \left| \frac{1}{d} \sum_{q=0}^{d-1} \gamma_q \omega^{q (k-\theta_x+\zeta_y)} \right|^2. 
\end{align}
Substituting the specific settings $x=y=i$, the phase term becomes $-\theta_i + \zeta_i = \frac{1}{2m}$. This yields:
\begin{align}
    \mathbb{P}\left(\widetilde{(A^1_i)_{\ket{\phi}}} = \widetilde{(B^1_i)_{\ket{\phi}}} + k \right) = d \left| \frac{1}{d} \sum_{q = 0}^{d-1} \gamma_q \cdot \omega^{q (k + \frac{1}{2m})} \right|^2. \label{eq:prob_indep}
\end{align}
Observing that Eq.~\eqref{eq:prob_indep} is independent of the measurement index $i$, the average over the $m$ settings simplifies directly:
\begin{align}
    \mathbb{P}(x_k) &= \frac{1}{m} \sum_{i=1}^{m} \mathbb{P}\left( \widetilde{(A^1_i)_{\ket{\phi}}} = \widetilde{(B)^1_i)_{\ket{\phi}}} + k \right) \nonumber = \frac{1}{m} \sum_{i=1}^{m} \left( d \left| \frac{1}{d} \sum_{q = 0}^{d-1} \gamma_q \cdot \omega^{q (k + \frac{1}{2m})} \right|^2 \right) \nonumber \\
    &= \frac{1}{m} \cdot m \cdot d \left| \frac{1}{d} \sum_{q = 0}^{d-1} \gamma_q \cdot \omega^{q (k + \frac{1}{2m})} \right|^2 \nonumber = d \left| \frac{1}{d} \sum_{q = 0}^{d-1} \gamma_q \cdot \omega^{q (k + \frac{1}{2m})} \right|^2.
\end{align}
\end{proof}

\begin{restatable}[]{lm}{OneMinusE} \label{lm:OneMinusE} ~\\
\[
\left| 1 - e^{\mathbf{i} \theta} \right|^2 = 2 (1 - \cos \theta) = 4 \sin^2 \left(\frac{\theta}{2} \right)
\]
\end{restatable}

\begin{proof}
\[
\begin{split}
\left| 1 - e^{\mathbf{i} \theta} \right|^2 &= (1 - e^{\mathbf{i} \theta}) (1 - e^{-\mathbf{i} \theta}) = 1 - e^{-\mathbf{i} \theta} - e^{\mathbf{i} \theta} + e^{\mathbf{i} \theta} e^{-\mathbf{i} \theta} = 2 - (e^{\mathbf{i} \theta} + e^{-\mathbf{i} \theta}) \\
 &= 2 - 2 \cos(\theta) \,\,\,\,\,\; \because e^{\mathbf{i} \theta} + e^{-\mathbf{i} \theta} = 2 \cos (\theta) \\
 &= 2 (1 - \cos (\theta)) = 2 (2 \sin^2 (\frac{\theta}{2})) \,\,\,\,\,\; \because 1 - \cos (\theta) = 2 \sin^2 (\frac{\theta}{2}) \\
 &= 4 \sin^2 \left(\frac{\theta}{2} \right)
\end{split}
\]
\end{proof}

\begin{restatable}[Finite Trigonometric Sums]{lm}{TrigIdentities} \label{lemma:trig_identities} ~\\
For any integer $d \ge 2$ and shift parameter $x \notin \{ \frac{\pi j}{d} \mid j \in \mathbb{Z} \}$, the following identities hold:
\begin{subequations}
    \begin{align}
        \sum_{k=0}^{d-1} \csc^2\left( \frac{\pi k}{d} + x \right) &= d^2 \csc^2(dx), \label{eq:csc2_sum} \\
        \sum_{k=0}^{d-1} \csc^4\left( \frac{\pi k}{d} + x \right) &= d^4 \csc^4(dx) - \frac{2}{3}d^2(d^2-1)\csc^2(dx), \label{eq:csc4_sum} \\
        \sum_{k=0}^{d-1} \cot^3\left( \frac{\pi k}{d} + x \right) &= d^3 \cot^3(dx) + d(d^2-1)\cot(dx). \label{eq:cot3_sum}
    \end{align}
\end{subequations}
\end{restatable}

\begin{proof}
The proof follows from the finite product identity for the sine function~\cite[Formula 1.392.1]{GR14}:
\begin{align}
    \sin(dx) &= 2^{d-1} \prod_{k=0}^{d-1} \sin\left(x + \frac{k\pi}{d}\right) \nonumber\\
    \ln(\sin dx) &= (d-1)\ln 2 + \sum_{k=0}^{d-1} \ln\left[\sin\left(x + \frac{k\pi}{d}\right)\right] \nonumber\\
    \frac{d}{dx} \ln(\sin dx) &= \sum_{k=0}^{d-1} \frac{d}{dx} \ln\left[\sin\left(x + \frac{k\pi}{d}\right)\right] \nonumber\\
    \frac{d \cos(dx)}{\sin(dx)} &= \sum_{k=0}^{d-1} \frac{\cos\left(x + \frac{k\pi}{d}\right)}{\sin\left(x + \frac{k\pi}{d}\right)} \nonumber\\
    d \cot(dx) &= \sum_{k=0}^{d-1} \cot\left(x + \frac{k\pi}{d}\right). \label{eq:base_identity}
\end{align}

\par\medskip
\noindent\textbf{Proof of Eq.~\eqref{eq:csc2_sum}.}
By differentiating both sides of Eq.~\eqref{eq:base_identity} with respect to $x$, and applying the derivative $\frac{d}{du}\cot u = -\csc^2 u$ along with the chain rule, we obtain:
\begin{equation}
    \sum_{k=0}^{d-1} \left[ -\csc^2\left( \frac{\pi k}{d} + x \right) \right] = d \cdot \left[ -d \csc^2(dx) \right] = -d^2 \csc^2(dx). \nonumber
\end{equation}
Multiplying by $-1$ yields the result.

\par\medskip
\noindent\textbf{Proof of Eq.~\eqref{eq:csc4_sum}.}
We differentiate Eq.~\eqref{eq:base_identity} three times with respect to $x$. Recall the third derivative identity $\frac{d^3}{du^3} \cot u = -6 \csc^4 u + 4 \csc^2 u$. Applying this operator, noting that the chain rule introduces a factor of $d^3$ on the right-hand side (combined with the initial $d$ to form $d^4$), we get:
\begin{equation}
    \sum_{k=0}^{d-1} \left( -6 \csc^4 \theta_k + 4 \csc^2 \theta_k \right) = d^4 \left( -6 \csc^4(dx) + 4 \csc^2(dx) \right), \nonumber
\end{equation}
where $\theta_k = \frac{\pi k}{d} + x$. Substituting Eq.~\eqref{eq:csc2_sum} for the $\sum \csc^2 \theta_k$ term and rearranging:
\begin{align}
    -6 \sum_{k=0}^{d-1} \csc^4 \theta_k &= -6 d^4 \csc^4(dx) + 4 d^4 \csc^2(dx) - 4 (d^2 \csc^2(dx)) \nonumber \\
    &= -6 d^4 \csc^4(dx) + 4 d^2(d^2-1) \csc^2(dx). \nonumber
\end{align}
Dividing by $-6$ gives the identity.

\par\medskip
\noindent\textbf{Proof of Eq.~\eqref{eq:cot3_sum}.}
Using the relation $\cot^3 u = \cot u \csc^2 u - \cot u$, we first evaluate the term $\sum \cot \theta_k \csc^2 \theta_k$ by differentiating Eq.~\eqref{eq:csc2_sum} with respect to $x$:
\begin{equation}
    \sum_{k=0}^{d-1} -2\csc^2\theta_k \cot\theta_k = d^2 \cdot \left[ -2d \csc^2(dx) \cot(dx) \right].
\end{equation}
Thus, $\sum \cot\theta_k \csc^2\theta_k = d^3 \cot(dx) \csc^2(dx)$. Substituting this back into the expansion, and recalling the base identity $\sum \cot\theta_k = d\cot(dx)$ (Eq.~\eqref{eq:base_identity}):
\begin{align}
    \sum_{k=0}^{d-1} \cot^3\theta_k &= \sum_{k=0}^{d-1} (\cot\theta_k \csc^2\theta_k - \cot\theta_k) \nonumber
    = d^3 \cot(dx) \csc^2(dx) - d \cot(dx) \nonumber \\
    &= d \cot(dx) \left[ d^2 (\cot^2(dx) + 1) - 1 \right] \nonumber
    = d^3 \cot^3(dx) + d(d^2-1)\cot(dx). \nonumber
\end{align}
\end{proof}

\section{Theorems and Their Proofs} \label{apx:ThmProofs}

\TightBound*

\begin{proof}
We know that $(U_1 \otimes U_2) \ket{\Phi_d} = (I \otimes U_2 U_1^T) \ket{\Phi_d}$ based on \cite{SATW17,SW22}.
Let $\ket{\phi'} = (U_1 \otimes U_2) \ket{\Phi_d} = (I \otimes U_2 U_1^T) \ket{\Phi_d}$.  

Suppose that $U_2 U_1^T = \sum_{j=0}^{d-1} c_j \ket{\lambda_j}\bra{\lambda_j}$ is the orthogonal decomposition of $U_2 U_1^T$, where $c_j \in \mathbb{C}$. Since $\ket{\Phi_d} = \frac{1}{\sqrt{d}} \sum_{j=0}^{d-1} \ket{\lambda_j} \ket{\lambda_j}^*$, we have
\[
\ket{\phi'} = \frac{1}{\sqrt{d}} \sum_{j=0}^{d-1} c_j \ket{\lambda_j} \ket{\lambda_j}^*.
\]
We can decompose $\ket{\phi'}$ in terms of two orthogonal vectors $\ket{\Phi_d}$ and $\ket{\Phi_d^\bot}$, i.e., $\bra{\Phi_d^\bot} \Phi_d \rangle = 0$ and $\ket{\Phi_d}, \ket{\Phi_d^\bot} \in \mathcal{S}$. 
Let $\ket{\phi'} = c_{\Phi_d} \ket{\Phi_d} + c_\bot \ket{\Phi_d^\bot}$ be the decomposition, where $c_{\Phi_d}, c_\bot \in \mathbb{C}$ and $|c_{\Phi_d}|^2 + |c_\bot|^2 = 1$.  Let us figure out $c_{\Phi_d}$.
\[
c_{\Phi_d} = \bra{\Phi_d} \phi' \rangle = \frac{1}{d} \sum_{j = 0}^{d-1} c_j = \frac{1}{d} \cdot \mathsf{Tr}(U_2 U_1^T)
\]
Recall that $D(U_1, U_2) = \sqrt{1 -  \left| \frac{1}{d} \mathsf{Tr}(U_1^T U_2) \right|^2} = \sqrt{1 -  \left| \frac{1}{d} \mathsf{Tr}(U_2 U_1^T) \right|^2}$. Then, we can conclude
\begin{equation} \label{Eq:Relation_D_CPhiD}
    D(U_1, U_2)^2 = 1 - |c_{\Phi_d}|^2.
\end{equation}

To simplify the notation, let us use $\mathcal{M}$ to denote $\sum_{x=1}^{m} \sum_{l=1}^{d-1} \left( A_x^l \otimes \overline{B}_x^l \right)$. Let us focus on $\Omega$ now. 
\[
\begin{split}
\Omega &= I_{d,m}(\ket{\phi'}) = \bra{\phi'} \mathcal{M} \ket{\phi'} = \left( c_{\Phi_d}^* \bra{\Phi_d} + c_\bot^* \bra{\Phi_d^\bot} \right) \mathcal{M} \left( c_{\Phi_d} \ket{\Phi_d} + c_\bot \ket{\Phi_d^\bot} \right) \\
 &= |c_{\Phi_d}|^2 \bra{\Phi_d} \mathcal{M} \ket{\Phi_d} + |c_\bot|^2 \bra{\Phi_d^\bot} \mathcal{M} \ket{\Phi_d^\bot} \\ 
  &= |c_{\Phi_d}|^2 \cdot I_{d,m}(\ket{\Phi_d}) + (1 - |c_{\Phi_d}|^2) \cdot I_{d,m}(\ket{\Phi_d^\bot}) \qquad \because |c_{\Phi_d}|^2 + |c_\bot|^2 = 1 \\ 
  &= |c_{\Phi_d}|^2 \cdot m(d-1) + (1 - |c_{\Phi_d}|^2) \cdot (-m) \qquad \qquad \because \mbox{Lemma~\ref{lm:CircuitEquiv} and Lemma~\ref{lm:IDM-M}}\\ 
  & = md \cdot |c_{\Phi_d}|^2 - m = md \left( 1 - D(U_1, U_2)^2 \right) - m \qquad \qquad \because \mbox{Eq.~\ref{Eq:Relation_D_CPhiD}}
\end{split}
\]
Therefore, we prove that $D(U_1, U_2) = \sqrt{1 - \frac{\Omega + m}{md}}$.
\end{proof}

\UfDotDot*

\begin{proof}
First of all, the state of $\ket{\vec{v}}$ is not changed by $\Ddot{U}_f$, because only the components $C_i$ and $C_i^{-1}$ involve $\ket{\vec{v}}$, but they never change its state.
Now, let us focus on the state of $\ket{F}$. It is modified by a multi-controlled-$Z$ gate, i.e., when all the clauses are satisfied (meaning that $f(\vec{v}) = 1$), the $Z$ gate is applied on $\ket{F}$. We know that $Z(\ket{0}) = \ket{0}$ and $Z(\ket{1}) = - \ket{1}$. Thus, we prove our claim.
\end{proof}

\FalseIdentity*

\begin{proof}
Based on our construction scheme, $\Ddot{U}_f$ has $(n+1)$ qubits. Thus, its dimension is $2^{(n+1)}$ by $2^{(n+1)}$. Since $f$ is unsatisfiable, $f(\vec{v}) = 0$ for all $\vec{v} \in \{0,1\}^n $. With Theorem~\ref{thm:UfDotDot-1} together, we have $\Ddot{U}_f (\ket{\vec{v}}\ket{F}) = \ket{\vec{v}}\ket{F}$ for all $\vec{v}$ and $\vec{F}$. Thus, $\Ddot{U}_f$ is the identity matrix.
\end{proof}

\IDMvsRatio*

\begin{proof}
By Theorem~\ref{thm:TightBound}, we know that
\[
\sqrt{1 - \widehat{I_{d,m}}} = D(\Ddot{U}_f, I) = \sqrt{1 -  \left| \frac{1}{2^{n+1}} \mathsf{Tr}(\Ddot{U}_f^T \cdot I) \right|^2} = \sqrt{1 -  \left| \frac{1}{2^{n+1}} \mathsf{Tr}(\Ddot{U}_f \cdot I^T) \right|^2}. 
\]
Thus, we have \[\widehat{I_{d,m}} = \left| \frac{1}{2^{n+1}} \mathsf{Tr}(\Ddot{U}_f) \right|^2.\] 
By Lemma~\ref{lm:FormulaTrace}, we know that $\mathsf{Tr}(\Ddot{U}_f) = 2^{n+1} - 2K$. Therefore, we prove that  
\[ \widehat{I_{d,m}} = \left( \frac{2^{n+1} - 2K}{2^{n+1}} \right)^2 = \left( 1 - \frac{K}{2^n} \right)^2 = (1 - r)^2. \]
\end{proof}

\StandardVarForNoSol*

\begin{proof}
If $K=0$, the formula $f$ has no solution, and $\Ddot{U}_f$ would be the identity matrix $I$. Thus, $\ket{\phi} = (\Ddot{U}_f \otimes I)\ket{\Phi_d} = (I \otimes I)\ket{\Phi_d} = \ket{\Phi_d} = \frac{1}{\sqrt{d}} \sum_{q=0}^{d-1} \ket{q}\ket{q}$. That is, $\gamma_q = \frac{1}{\sqrt{d}}$ for all $q \in \{0, 1, \ldots, d-1 \}$. By Lemma~\ref{lm:general_second_moment}, we have
\[
\mathbb{E}(X_{d,m,K}^2) = \sum_{k=0}^{d-1} (2\alpha_k)^2 \cdot \mathbb{P}(x_k), \quad \mbox{where} \quad \mathbb{P}(x_k) = d \left| \frac{1}{d} \sum_{q = 0}^{d-1} \gamma_q \cdot \omega^{q (k + \frac{1}{2m})} \right|^2.
\]
Let us simplify $\mathbb{P}(x_k)$ with $\gamma_q =\frac{1}{\sqrt{d}}$, $\xi = \frac{\pi}{2m}$, $\phi_k = \frac{\pi}{d}(k + \frac{1}{2m})$, and $\omega = e^{\frac{2 \pi \mathbf{i}}{d}}$.
\begin{align}
    \mathbb{P}(x_k) &= d \left| \frac{1}{d} \sum_{q = 0}^{d-1} \gamma_q \cdot \omega^{q (k + \frac{1}{2m})} \right|^2 
    = d \left| \frac{1}{d\sqrt{d}} \sum_{q=0}^{d-1} (e^{\mathbf{i} 2\phi_k})^q \right|^2 \nonumber \\
    &= \frac{1}{d^2} \left| \frac{1 - (e^{\mathbf{i} 2\phi_k})^d}{1 - e^{\mathbf{i} 2\phi_k}} \right|^2 \quad \quad \quad  \because \sum_{q=0}^{n-1} r^q = \frac{1-r^n}{1-r} \mbox{ , and let } r = e^{\mathbf{i} 2 \phi_k} \nonumber \\
    &= \frac{1}{d^2} \frac{|1 - e^{\mathbf{i} 2 d \phi_k}|^2}{|1 - e^{\mathbf{i} 2 \phi_k}|^2} \quad \quad \quad \quad \because |a/b|^2 = |a|^2/|b|^2 \nonumber \\
    &= \frac{1}{d^2} \frac{|1 - e^{\mathbf{i} 2 (\pi k + \xi)}|^2}{|1 - e^{\mathbf{i} 2 \phi_k}|^2} \quad \quad \quad \because d\phi_k = d \cdot \frac{\pi}{d}(k + \frac{1}{2m}) = \pi k + \xi \nonumber \\
    &= \frac{1}{d^2} \frac{|1 - e^{\mathbf{i} 2 \xi}|^2}{|1 - e^{\mathbf{i} 2 \phi_k}|^2} \quad \quad \quad \quad \quad \because e^{\mathbf{i} 2(\pi k + \xi)} = e^{\mathbf{i} 2\pi k} \cdot e^{\mathbf{i} 2\xi} = 1 \cdot e^{\mathbf{i} 2\xi} \nonumber \\
    &= \frac{1}{d^2} \frac{4\sin^2 \xi}{4\sin^2 \phi_k} \quad \quad \quad \quad \because \text{Lemma~\ref{lm:OneMinusE}} \nonumber \\
    &= \frac{1}{d^2} \frac{\sin^2 \xi}{\sin^2 \phi_k}. \label{eq:Pk_detailed_proof}
\end{align}

Recalling the definition of $\alpha_k$ in Lemma~\ref{lm:NormBellExp}, we further simplify it with $\xi = \frac{\pi}{2m}$ and $\phi_k = \frac{\pi}{d}(k + \frac{1}{2m})$, as follows:
\begin{equation} \label{eq:alpha_k_definition}
    \alpha_k = \frac{1}{2d} \tan\left(\frac{\pi}{2m}\right) \cot \left( \frac{\pi}{d} (k + \frac{1}{2m}) \right) = \frac{\tan \xi \cot \phi_k}{2d}
\end{equation}
Now, let us rewrite $\mathbb{E}[X_{d,m,0}^2]$ with the conditions of $(2\alpha_k)^2 = \frac{\tan^2 \xi \cot^2 \phi_k}{d^2}$ and $\mathbb{P}(x_k) = \frac{1}{d^2} \frac{\sin^2 \xi}{\sin^2 \phi_k}$, we have

\begin{align}
    \mathbb{E}[X_{d,m,0}^2] &= \sum_{k=0}^{d-1} \mathbb{P}(x_k) (2\alpha_k)^2 = \sum_{k=0}^{d-1} \left( \frac{\sin^2 \xi}{d^2 \sin^2 \phi_k} \right) \left( \frac{\tan^2 \xi \cot^2 \phi_k}{d^2} \right) \nonumber \\
    &= \frac{\tan^2 \xi \sin^2 \xi}{d^4} \sum_{k=0}^{d-1} \csc^2 \phi_k \cot^2 \phi_k \quad \quad \because \frac{1}{\sin \phi_k} = \csc \phi_k \nonumber \\
    &= \frac{\tan^2 \xi \sin^2 \xi}{d^4} \left[ \sum_{k=0}^{d-1} \csc^4 \phi_k - \sum_{k=0}^{d-1} \csc^2 \phi_k \right] \quad \quad \because \cot^2 \phi_k = \csc^2 \phi_k - 1 \nonumber \\
    &= \frac{\tan^2 \xi \sin^2 \xi}{d^4} \bigg[ \left( d^4 \csc^4 \xi - \tfrac{2}{3}d^2(d^2-1)\csc^2 \xi \right) - d^2 \csc^2 \xi \bigg]\nonumber \\
    & \quad \hspace{3cm} \because \mbox{By Eqs.~\eqref{eq:csc2_sum}, \eqref{eq:csc4_sum} in Lemma.~\ref{lemma:trig_identities} with } x = \xi/d \nonumber \\
    &= \frac{\tan^2 \xi}{d^4} \left[ d^4 \csc^2 \xi - \frac{2}{3}d^2(d^2-1) - d^2 \right] \quad \quad \because \sin^2 \xi \csc^2 \xi = 1 \nonumber \\
    &= \tan^2 \xi \csc^2 \xi - \frac{2(d^2-1)}{3d^2}\tan^2 \xi - \frac{1}{d^2}\tan^2 \xi \nonumber \\
    &= (1 + \tan^2 \xi) - \frac{2d^2-2}{3d^2}\tan^2 \xi - \frac{3}{3d^2}\tan^2 \xi \quad \because \tan^2 \xi \csc^2 \xi = 1 + \tan^2 \xi \nonumber \\
    &= 1 + \tan^2 \xi \left( \frac{3d^2 - (2d^2 - 2) - 3}{3d^2} \right) = 1 + \frac{d^2-1}{3d^2} \tan^2 \xi. \nonumber
\end{align}

With the fact that $\mathbb{E}[X_{d,m,0}] = 1$, we finish our proof in the following.
\begin{align}
    \sigma^2(X_{d,m,0}) &= \mathbb{E}[X_{d,m,0}^2] - \mathbb{E}[X_{d,m,0}] = \left(1 + \frac{d^2-1}{3d^2} \tan^2 \xi \right) - 1 = \frac{d^2-1}{3d^2} \tan^2 \xi \nonumber \\
    \sigma(X_{d,m,0}) &= \frac{1}{d} \sqrt{\frac{d^2-1}{3}} \tan\left(\frac{\pi}{2m}\right). \nonumber
\end{align} 
\end{proof}

\StandardVarForOneSol*

\begin{proof}
When $K = 1$, according to Lemma~\ref{lemma:max_variance_K1}, the maximum variance occurs at $q^* = 0$. Thus, the coefficients of $\ket{\phi}$ would be $\gamma_0 = -1/\sqrt{d}$ and $\gamma_q = 1/\sqrt{d}$ for $q \in \{1,2, \ldots, d-1 \}$. 
Let us simplify $\mathbb{P}(x_k)$ further with $\xi = \frac{\pi}{2m}$, $\phi_k = \frac{\pi}{d}(k + \frac{1}{2m})$, and $\omega = e^{\frac{2 \pi \mathbf{i}}{d}}$.
\begin{align}
    \mathbb{P}(x_k) &= d \left| \frac{1}{d} \sum_{q = 0}^{d-1} \gamma_q \cdot \omega^{q (k + \frac{1}{2m})} \right|^2 \\
    &= \frac{1}{d} \left|  \left( \sum_{q = 0}^{d-1} \frac{1}{\sqrt{d}} \cdot \omega^{q (k + \frac{1}{2m})} \right) - 2 \left( \frac{1}{\sqrt{d}} \right) \omega^0 \right|^2 \nonumber \\
    &= \frac{1}{d^2} \left| \left(\sum_{q=0}^{d-1} r^q\right) - 2r^0 \right|^2 \quad \quad \mbox{, where} \quad r = e^{\mathbf{i} 2\phi_k} \nonumber \\
    &= \frac{1}{d^2} \left| \frac{1-r^d}{1-r} - 2 \right|^2 \\
    &= \frac{1}{d^2} \left| \frac{2r - (1+r^d)}{1-r} \right|^2 \nonumber \\
    &= \frac{1}{d^2} \frac{|2e^{i 2\phi_k} - (1 + e^{i 2 d \phi_k})|^2}{|1 - e^{i 2\phi_k}|^2} \quad \because |a/b|^2 = |a|^2/|b|^2 \nonumber \\
    &= \frac{1}{d^2} \frac{|2e^{i 2\phi_k} - (1 + e^{i 2(\pi k + \xi)})|^2}{|1 - e^{i 2\phi_k}|^2} \quad \because d\phi_k = d \cdot \frac{\pi}{d}(k + \frac{1}{2m}) = \pi k + \xi \nonumber \\
    &= \frac{1}{d^2} \frac{|2e^{i 2\phi_k} - (1 + e^{i 2\xi})|^2}{|1 - e^{i 2\phi_k}|^2} \quad \because e^{i 2(\pi k + \xi)} = e^{i 2\pi k} e^{i 2\xi} = 1 \cdot e^{i 2\xi} \nonumber \\
    &= \frac{1}{d^2} \frac{|2e^{i 2\phi_k} - e^{i\xi}(e^{-i\xi} + e^{i\xi})|^2}{4\sin^2\phi_k} \quad \because |1-e^{i\theta}|^2 = 4\sin^2(\theta/2) \text{ and factor out } e^{i\xi} \nonumber \\
    &= \frac{1}{d^2} \frac{|2e^{i 2\phi_k} - 2e^{i\xi}\cos\xi|^2}{4\sin^2\phi_k} \quad \because e^{i\xi} + e^{-i\xi} = 2\cos\xi \nonumber \\
    &= \frac{|e^{i 2\phi_k} - e^{i\xi}\cos\xi|^2}{d^2 \sin^2 \phi_k}. \quad \because \text{factor out constant 2 and cancel } 4/4
\end{align}
Expanding the squared modulus in the numerator yields:
\begin{align}
    \mathbb{P}(x_k) &= \frac{1}{d^2 \sin^2 \phi_k} (e^{i 2\phi_k} - e^{i\xi}\cos\xi)(\overline{e^{i 2\phi_k} - e^{i\xi}\cos\xi}) \quad \because \text{Property } |z|^2 = z\bar{z} \nonumber \\
    &= \frac{1}{d^2 \sin^2 \phi_k} (e^{i 2\phi_k} - e^{i\xi}\cos\xi)(e^{-i 2\phi_k} - e^{-i\xi}\cos\xi) \quad \because \overline{e^{i\theta}} = e^{-i\theta} \text{ and } \cos\xi \in \mathbb{R} \nonumber \\
    &= \frac{1}{d^2 \sin^2 \phi_k} (e^{i 2\phi_k}e^{-i 2\phi_k} + e^{i\xi}e^{-i\xi}\cos^2\xi - e^{i 2\phi_k}e^{-i\xi}\cos\xi - e^{-i 2\phi_k}e^{i\xi}\cos\xi) \nonumber \\
    &= \frac{1}{d^2 \sin^2 \phi_k} (1 + \cos^2\xi - \cos\xi [e^{i(2\phi_k-\xi)} + e^{-i(2\phi_k-\xi)}]) \quad \because e^{i\theta}e^{-i\theta} = 1 \nonumber \\
    &= \frac{1}{d^2 \sin^2 \phi_k} (1 + \cos^2\xi - 2\cos\xi \cos(2\phi_k - \xi)) \quad \because e^{i\theta} + e^{-i\theta} = 2\cos\theta \nonumber \\
    &= \frac{1}{d^2 \sin^2 \phi_k} \left( 1 + \cos^2\xi - 2\cos\xi [ \cos(2\phi_k)\cos\xi + \sin(2\phi_k)\sin\xi ] \right) \nonumber \\
    & \quad \because \cos(A-B) = \cos A \cos B + \sin A \sin B \nonumber \\
    &= \frac{1}{d^2 \sin^2 \phi_k} \left( 1 + \cos^2\xi - 2\cos^2\xi \cos(2\phi_k) - 2\cos\xi \sin\xi \sin(2\phi_k) \right) \nonumber \\
    &= \frac{1}{d^2 \sin^2 \phi_k} \left( 1 + \cos^2\xi - 2\cos^2\xi (1 - 2\sin^2\phi_k) - \sin(2\xi) (2\sin\phi_k \cos\phi_k) \right) \nonumber \\
    & \quad \because \cos(2\phi_k) = 1-2\sin^2\phi_k \text{ and } 2\sin\xi\cos\xi = \sin(2\xi) \text{ and } \sin(2\phi_k) = 2\sin\phi_k\cos\phi_k \nonumber \\
    &= \frac{1}{d^2 \sin^2 \phi_k} \left( 1 + \cos^2\xi - 2\cos^2\xi + 4\cos^2\xi \sin^2\phi_k - 2\sin(2\xi) \sin\phi_k \cos\phi_k \right) \nonumber \\
    &= \frac{1}{d^2 \sin^2 \phi_k} \left( (1 - \cos^2\xi) + 4\cos^2\xi \sin^2\phi_k - 2\sin(2\xi) \sin\phi_k \cos\phi_k \right) \nonumber \\
    &= \frac{1}{d^2 \sin^2 \phi_k} \left( \sin^2\xi + 4\cos^2\xi \sin^2\phi_k - 2\sin(2\xi) \sin\phi_k \cos\phi_k \right) \quad \because 1 - \cos^2\xi = \sin^2\xi \nonumber \\
    &= \frac{1}{d^2} \left( \frac{\sin^2\xi}{\sin^2\phi_k} + \frac{4\cos^2\xi\sin^2\phi_k}{\sin^2\phi_k} - \frac{2\sin(2\xi)\sin\phi_k\cos\phi_k}{\sin^2\phi_k} \right) \nonumber \\
    &= \frac{1}{d^2} \left[ \sin^2\xi \csc^2\phi_k + 4\cos^2\xi - 2\sin(2\xi)\cot\phi_k \right]. \quad \because \csc \phi = \frac{1}{\sin \phi} \text{ and } \cot \phi = \frac{\cos \phi}{\sin \phi} \nonumber
\end{align}
Recalling the definition in Eq.~\eqref{eq:alpha_k_definition}, we have $(2\alpha_k)^2 = \frac{\tan^2 \xi}{d^2} \cot^2 \phi_k$. Substituting this expression along with the simplified $\mathbb{P}(x_k)$ into the second moment formula yields:

\begin{align}
    \mathbb{E}[X_{d,m,1}^2] &= \sum_{k=0}^{d-1} \mathbb{P}(x_k) (2\alpha_k)^2 \nonumber \\
    &= \frac{\tan^2\xi}{d^4} \sum_{k=0}^{d-1} \left[ \sin^2\xi \csc^2\phi_k + 4\cos^2\xi - 2\sin(2\xi)\cot\phi_k \right] \cot^2\phi_k \nonumber \\
    &= \frac{\tan^2\xi}{d^4} \sum_{k=0}^{d-1} \left[ \sin^2\xi (\csc^4\phi_k - \csc^2\phi_k) + 4\cos^2\xi (\csc^2\phi_k - 1) - 2\sin(2\xi)\cot^3\phi_k \right] \nonumber \\
    & \quad \because \cot^2\phi = \csc^2\phi - 1 \text{ and } \csc^2\phi\cot^2\phi = \csc^4\phi - \csc^2\phi \nonumber \\
    &= \frac{\tan^2\xi}{d^4} \left[ \sin^2\xi (\mathcal{S}_{csc4} - \mathcal{S}_{csc2}) + 4\cos^2\xi (\mathcal{S}_{csc2} - d) - 2\sin(2\xi) \mathcal{S}_{cot3} \right], \label{eq:moment_sum_expanded}
\end{align}
where $\mathcal{S}_{csc2}, \mathcal{S}_{csc4}, \mathcal{S}_{cot3}$ are the trigonometric sums from Lemma~\ref{lemma:trig_identities} with $x = \xi/d$. Substituting these identities into Eq.~\eqref{eq:moment_sum_expanded} and performing the algebraic reduction:
\begin{align}
    \mathbb{E}[X_{d,m,1}^2] &= \frac{1}{d^4} \left[ d^4 \sec^2\xi - \left(\tfrac{2}{3}d^4 + \tfrac{1}{3}d^2\right) \tan^2\xi + (4d^2 - 4d\sin^2\xi) + (-4d^3 + 4d\sin^2\xi) \right] \nonumber \\
    &= \frac{1}{d^4} \left[ \frac{d^2(d^2-1)}{3}\tan^2\xi + (d^4 - 4d^3 + 4d^2) \right]. \nonumber
\end{align}
With $\mathbb{E}[X_{d,m,1}] = (1 - \frac{2}{d})^2$, the variance is derived as:
\begin{align}
    \sigma^2(X_{d,m,1}) &= \mathbb{E}[X_{d,m,1}^2] - (\mathbb{E}[X_{d,m,1}])^2 \nonumber \\
    &= \frac{1}{d^4} \left[ \frac{d^2(d^2-1)}{3}\tan^2\xi + (d^4 - 4d^3 + 4d^2) - (d-2)^4 \right] \nonumber \\
    &= \frac{1}{d^4} \left[ \frac{d^2(d^2-1)}{3}\tan^2\xi + (d^4 - 4d^3 + 4d^2) - (d^4 - 8d^3 + 24d^2 - 32d + 16) \right] \nonumber \\
    &= \frac{1}{d^4} \left[ \frac{d^2(d^2-1)}{3}\tan^2\xi + (4d^3 - 20d^2 + 32d - 16) \right] \nonumber \\
    &= \frac{1}{d^4} \left[ \frac{d^2(d^2-1)}{3}\tan^2\xi + 4(d-1)(d-2)^2 \right]. \nonumber
\end{align}
Taking the square root of both sides, we prove that
\begin{equation}
    \sigma(X_{d,m,1}) = \frac{1}{d^2} \sqrt{ \frac{d^2(d^2-1)}{3}\tan^2\left(\frac{\pi}{2m}\right) + 4(d-1)(d-2)^2 }.
\end{equation}
\end{proof}

\ShannonExpansion*

\begin{proof}
(1) $f$ is $f_\bot$ $\iff$ $f_{[v_i]} = (f_{\bot})_{v_i \leftarrow 0} \vee (f_{\bot})_{v_i \leftarrow 1} = f_\bot \vee f_\bot$ $\iff$ $f_{v_i}$ is $f_\bot$.

(2) From Fig.~\ref{fig:Shannon_Expansion}, we understand that Shannon expansion is to merge the two sub solution spaces (each of which has $2^{n-1}$ elements) into one with only $2^{n-1}$ elements. When the merge is performed, solutions to $f$ have two cases. Let $\ket{v_i}\ket{\vec{x}}$ be the general form of the assignments to $f$, where $\vec{x} \in \{0,1\}^{|V \setminus \{v_i\}|}$ be an assignment in the subspace without $v_i$. The first case is that both assignments $\ket{0}\ket{\vec{x}}$ and $\ket{1}\ket{\vec{x}}$ are solutions to $f$. When the merge is performed, $\ket{0}\ket{\vec{x}}$ will cancel out $\ket{1}\ket{\vec{x}}$, which we consider as the worst case, marked in red color in Fig.~\ref{fig:Shannon_Expansion}. The second case is that either $\ket{0}\ket{\vec{x}}$ or $\ket{1}\ket{\vec{x}}$ is a solution to $f$, but not both. In this case, no solution will be canceled out, which we consider as the best case, marked in green color in Fig.~\ref{fig:Shannon_Expansion}. Assume formula $f$ has $k$ solutions, i.e., $r_f = \frac{k}{2^n}$.  One can easily check that if all solutions hit the worst case, then $r_{f_{[v_i]}} = \frac{k/2}{2^{n-1}} = \frac{k}{2^{n}} = r_f$. If all solutions hit the best case, then $r_{f_{[v_i]}} = \frac{k}{2^{n-1}} = 2 \cdot \frac{k}{2^{n}} = 2 \cdot r_f$. Thus, we prove that $r_f \leq r_{f_{[v_i]}} \leq 2 \cdot r_f$.

\begin{figure}[tb]
\centering
\begin{minipage}{0.56\linewidth}
\centering
\includegraphics[width=0.7\linewidth]{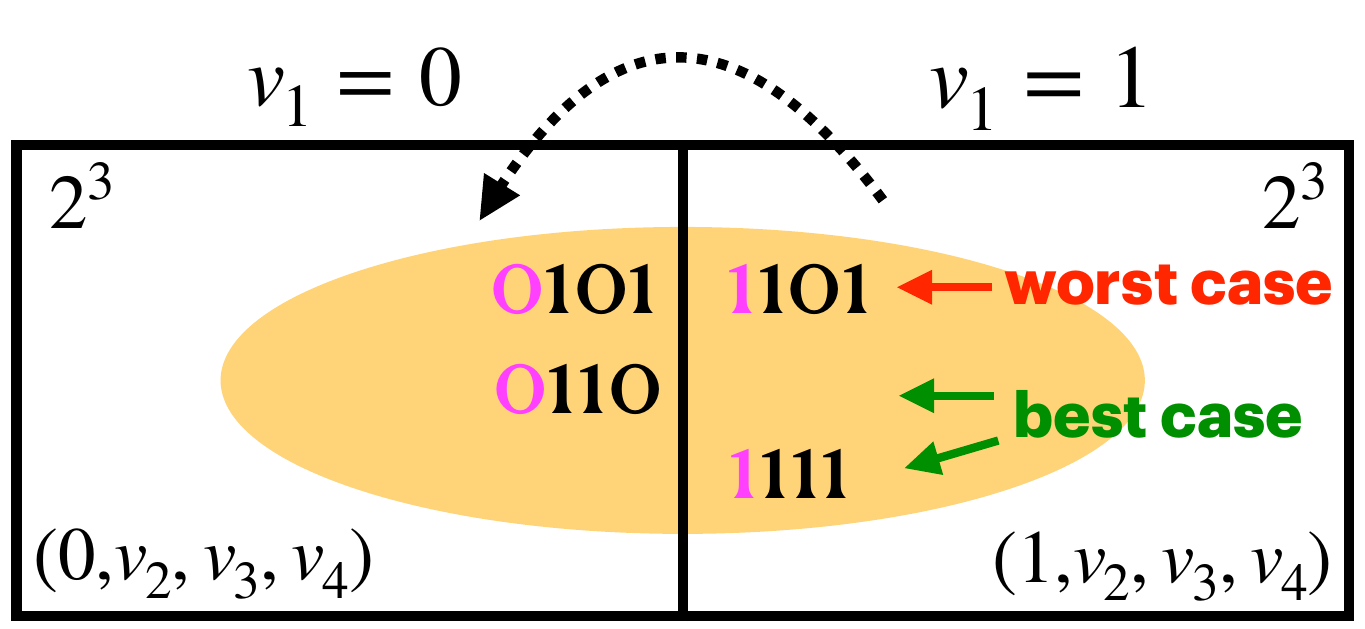}
\end{minipage}
~
\hspace{1mm}
~
\begin{minipage}{0.38\linewidth}
\centering
\includegraphics[width=0.5\linewidth]{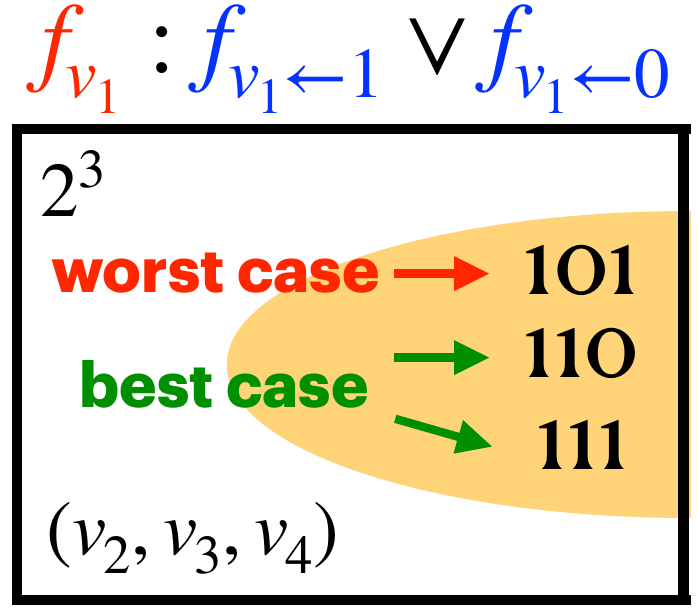}
\end{minipage}
\\[1mm]
\begin{minipage}{0.56\linewidth}
\centering
(a) before, $r = \frac{4}{2^4}$
\end{minipage}
~
\begin{minipage}{0.38\linewidth}
\centering
(b) after, $r = \frac{3}{2^3}$
\end{minipage}
\caption{Shannon Expansion for Formula $f$ over $V = \{ v_1, v_2, v_3, v_4 \}$}
\label{fig:Shannon_Expansion}
\end{figure}

(3) Assume formula $f$ has $2^t$ solutions. If we apply Shannon expansion recursively on $f$ for $n$ times, the changing of solution ratios in the worst case would be as follows. Notice that the red color indicates the turning point (when the numerator becomes $1$). After this point, it can be easily verified that only the best case can happen, i.e., the solution ratio will be doubled for $(n-k)$ iterations and then hits $1$, and the final result formula $F_V$ would be $f_\top$.
\[
\frac{2^t}{2^n} \rightarrow \frac{2^{t-1}}{2^{n-1}} \rightarrow \frac{2^{t-2}}{2^{n-2}} \rightarrow \cdots \rightarrow {\color{red} \frac{2^{0}}{2^{n-t}} } \rightarrow {\color{blue} \frac{2^{0}}{2^{(n-t)-1}} } \rightarrow {\color{blue} \frac{2^{0}}{2^{(n-t)-2}} } \rightarrow \cdots \rightarrow {\color{blue} \frac{2^{0}}{2^{0}} }
\]
If formula $f$ has no solutions, one can easily check that the numerators of the above sequence will be all $0$, and the final result formula would be $f_\bot$.
\end{proof}

\ShanExpCircuitCorrect*

\begin{proof}
First of all, the state of $\ket{\vec{v}}$ is not changed by $\Ddot{U}_{f_{[H]}}$, because only the components $U_{f_{x_1 x_2 \cdots x_h}}$ and $U_{f_{x_1 x_2 \cdots x_h}}$ involve $\ket{\vec{v}}$, but they never change its state.
Now, let us focus on the state of $\ket{F}$. There are two cases: (1) the input value of \ket{F} is $0$, and (2) the input value of \ket{F} is $1$.

Case~(1): $\ket{F} = 0$. The qubit $\ket{F}$ is modified only by a $Z$-gate, followed by a multi-controlled-$Z$-gate. We know that $Z(\ket{0}) = \ket{0}$. That is, whether the $Z$-gates are applied or not, the state of $\ket{F}$ is not changed. Thus, the final state would still be $\ket{\vec{v}} \ket{0} = (-1)^0 \cdot \ket{\vec{v}} \ket{0} = (-1)^{f_{[H]}(\vec{v}) \wedge 0} \ket{\vec{v}} \ket{0}$.

Case~(2): $\ket{F} = 1$. Initially, a $Z$-gate is applied on $\ket{F}$, and we know that $Z(\ket{1}) = -\ket{1}$. Thus, a ``$-1$'' phase is obtained in the beginning. Then, a multi-controlled $Z$-gate is applied on $\ket{F}$, i.e., the $Z$-gate is applied on $\ket{F}$ when each qubit $\ket{F_{x_1 x_2 \cdots x_h}}$ is $0$. Notice that qubit $\ket{F_{x_1 x_2 \cdots x_h}}$ indicates the truth value of $f_{x_1 x_2 \cdots x_h}(\vec{v})$. That is, if $f_{H}(\vec{v}) = 0$, the $Z$-gate is applied on $\ket{F}$ causing another ``$-1$'' phase, which cancels out the initial ``$-1$'' phase; otherwise, the $Z$-gate is not applied, and the initial ``$-1$'' phase remains. Thus, we have $\Ddot{U}_{f_{[H]}}(\ket{\vec{v}}\ket{1}) = (-1)^{f_{[H]}(\vec{v})} \ket{\vec{v}} \ket{1}$.
\end{proof}

\QSATCorrectness*

\begin{proof}
Let $\mathcal{E}$ denote the event that Algorithm~\ref{alg:QSAT} returns an incorrect result. The algorithm performs a sequence of statistical tests on the iteratively updated formula $\mathcal{F}_t$. Let $\mathcal{E}_t$ be the event that a statistical misclassification occurs at iteration $t$.

The global error event $\mathcal{E}$ is the union of all local error events $\mathcal{E} = \bigcup_{t=0}^{|V|} \mathcal{E}_t$. We analyze the local error $\mathcal{E}_t$ by considering the extremal cases of the solution count $K$:
\begin{itemize}
    \item \textbf{Type I Error (False Positive):} Occurs when $\mathcal{F}_t$ is UNSAT ($K=0$) but is misclassified as SAT. We utilize the standard deviation $\sigma_0$ (derived in Theorem~\ref{thm:StandardVar4NoSol}) to bound this error.
    \item \textbf{Type II Error (False Negative):} Occurs when $\mathcal{F}_t$ is SAT ($K \ge 1$) but is misclassified as UNSAT. Since $K=1$ represents the case most proximal to the UNSAT boundary (the lowest solution ratio), it constitutes the \emph{worst-case} scenario for Type II errors. We utilize $\sigma_1$ (derived in Theorem~\ref{thm:StandardVar4OneSol}) as a conservative bound for all $K \ge 1$.
\end{itemize}

In accordance with the \textbf{Bonferroni Correction} principle~\cite{Bonferroni36}, we configure the per-iteration significance level to be $\alpha = \frac{1-L}{|V|+1}$. In Lines 7--8 of Algorithm~\ref{alg:QSAT}, the thresholds $\epsilon_0$ and $\epsilon_1$ are computed by solving:
\begin{equation}
    2 \left( 1 - \Phi\left( \frac{\epsilon \sqrt{s}}{\sigma} \right) \right) = \alpha
\end{equation}
By constructing the thresholds using the extremal standard deviations $\sigma_0$ and $\sigma_1$, we guarantee that $\Pr(\mathcal{E}_t) \le \alpha$ for any number of solutions $K$. Specifically, if $K > 1$, the expected measurement $\mu_t$ moves further away from the decision boundary, rendering the error probability strictly lower than the bound established for $K=1$.

To derive the global reliability, we apply the \textbf{Union Bound} (Boole's inequality), which states that the probability of a union of events is at most the sum of their individual probabilities.
\begin{equation}
\Pr(\mathcal{E}) = \Pr\left( \bigcup_{t=0}^{|V|} \mathcal{E}_t \right) \le \sum_{t=0}^{|V|} \Pr(\mathcal{E}_t) \le \sum_{t=0}^{|V|} \frac{1-L}{|V|+1} = 1 - L
\end{equation}

Thus, the probability of returning the correct result is $\Pr(\text{Correct}) = 1 - \Pr(\mathcal{E}) \ge L$, confirming the algorithm satisfies the global confidence level $L$.
\end{proof}

\ExpectedIterations*

\begin{proof}
Let $N = 2^{|V|}$ be the search space size, and $a \in (1, 2]$ be the amplification factor for the solution ratio $r$ due to Shannon expansion (c.f.~Theorem~\ref{thm:Shannon_Expansion}{.2}). According to Algorithm~\ref{alg:QSAT}, the Shannon-expanded formula in the $t$-th iteration is classified as SAT if the observation $\mu_t \le 1 - \epsilon_0$. To reliably separate the confidence intervals of SAT and UNSAT instances, the expected value of $\mu_t$ must be bounded. Here, we introduce a global constant $\epsilon_{max}$. To be conservative, $\epsilon_{max} = \max_{t \ge 0} (\epsilon_{0,t} + \epsilon_{1,t})$, where $\epsilon_{0,t}$ and $\epsilon_{1,t}$ are the values of $\epsilon_0$ and $\epsilon_1$ solved in the $t$-th iteration, respectively. That is, it is required that $\mathbb{E}[\mu_t] \le 1 - \epsilon_{max}$.

By Theorem \ref{thm:IDMvsRatio}, the expected result in the $t$-th iteration corresponding to the solution ratio $r_t = \frac{K \cdot a^t}{N}$ is $\mathbb{E}[\mu_t] = (1 - r_t)^2$. Since $(1 - r_t)^2 \le 1 - r_t$ for $r_t \in [0, 1]$, a sufficient condition is $1 - r_t \le 1 - \epsilon_{max}$, which simplifies to $r_t \ge \epsilon_{max}$. Substituting $r_t$ in the inequality yields the required expansion steps $t$ for $K \ge 1$:
\begin{equation}
    \frac{K \cdot a^t}{N} \ge \epsilon_{max} \implies a^t \ge \frac{N \epsilon_{max}}{K} \implies t = \left\lceil \log_a \left( \frac{N \epsilon_{max}}{K} \right) \right\rceil. \label{eq:IterationNum_t}
\end{equation}

For UNSAT instances ($K=0$), Algorithm~\ref{alg:QSAT}, in the worst case, performs Shannon expansions for $|V|$ times, reaching the maximum value of $t = |V| = \log_2 N$. 
Since the probability of a formula with $K$ solutions is $\mathbb{P}(K) = \frac{1}{2^N}\binom{N}{K}$, we can calculate the expected number of Shannon expansions $\mathbb{E}[t]$ as follows:
\begin{align}
    \mathbb{E}[t] &= \mathbb{P}(0) \cdot \log_2 N + \sum_{K=1}^{N} \mathbb{P}(K)\cdot \left\lceil \log_a \left( \frac{N \epsilon_{max}}{K} \right) \right\rceil \nonumber \\
    &= \frac{\log_2 N}{2^N} + \sum_{K=1}^{N} \frac{1}{2^N} \binom{N}{K} \log_a \left( \frac{N \epsilon_{max}}{K} \right). \label{eq:expected_t_sum}
\end{align}

For large $N$, the binomial distribution is strongly dominated by the peak at $K = N/2$. We approximate the summation by evaluating the logarithmic term at this dominant point:
\begin{align}
    &\sum_{K=1}^{N} \frac{1}{2^N} \binom{N}{K} \log_a \left( \frac{N \epsilon_{max}}{K} \right) \nonumber \\
    &\approx \log_a \left( \frac{N \epsilon_{max}}{N/2} \right) \sum_{K=1}^{N} \frac{1}{2^N} \binom{N}{K} \quad \quad \because \text{Peak at } K = N/2 \nonumber \\
    &\approx \log_a (2\epsilon_{max}) \cdot 1 = \log_a(2\epsilon_{max}). \quad \quad \because \sum_{K=0}^N \frac{\binom{N}{K}}{2^N} = 1 \label{eq:dominant_approx_t}
\end{align}

Substituting Eq.~\ref{eq:dominant_approx_t} back into Eq.~\ref{eq:expected_t_sum}, we have
\begin{equation}
    \mathbb{E}[t] \approx \frac{\log_2 N}{2^N} + \log_a(2\epsilon_{max}). \nonumber
\end{equation}

As $N \to \infty$, the term $\frac{\log_2 N}{2^N} \to 0$. Therefore, the expected number of expansion steps converges to $\mathbb{E}[t] \approx \log_a(2\epsilon_{max})$. Since $\epsilon_{max}$ and $a$ are algorithm parameters determined by the fixed configurations, the expected depth is effectively bounded by $\mathcal{O}(1)$.
\end{proof}

\ExpectedExecTime*

\begin{proof}
Consider a Boolean formula $f$ with $K$ solutions. According to Eq.~\ref{eq:IterationNum_t}, since the number of iterations $t$, required by Algorithm~\ref{alg:QSAT} to solve $f$, depends on $K$, we use the symbol $t_K$ in the following. 
The total execution time $T$ to solve $f$ would be $T = s \cdot t_K \cdot \rho$, where $s$ is the number of samplings for each iteration, $t_K$ is the number of iterations, and $\rho$ is the execution time for each sampling. Since $\rho$ would be doubled when a Shannon expansion is performed, we use $\rho = c \cdot 2^{t_K}$ to be conservative, where $c$ is a constant representing the initial sampling time.

The expected execution time $\mathbb{E}[T]$ over all possible solution counts $K$ is directly formulated as $\mathbb{E}[T] = \sum_{K=0}^{N} \mathbb{P}(K) \cdot \left( s \cdot t_K \cdot (c \cdot 2^{t_K}) \right)$. As established in Theorem \ref{thm:expected_iterations}, the distribution $\mathbb{P}(K)$ is overwhelmingly concentrated at the mean $K = N/2$. This causes the contribution of tail cases to vanish asymptotically, as the exponential decay of the binomial tail strictly dominates the polynomial growth of the execution time. Thus, $\mathbb{E}[T]$ would be dominated by this peak so that we have
\begin{equation}
    \mathbb{E}[T] \approx s \cdot t_{N/2} \cdot (c \cdot 2^{t_{N/2}}).
\end{equation}

Substituting the dominant $t_{N/2} = \log_a(2\epsilon_{max})$ into the expectation yields:
\begin{align}
    \mathbb{E}[T] &\approx s \cdot \log_a(2\epsilon_{max}) \cdot c \cdot 2^{\log_a(2\epsilon_{max})} \nonumber \\
    &= s \cdot c \cdot \underbrace{ \log_a(2\epsilon_{max}) \cdot (2\epsilon_{max})^{\log_a 2} }_{\text{Constant bounded by the fixed parameters}} \quad \quad \because x^{\log_y z} = z^{\log_y x} \nonumber \\
    &= s \cdot O(1). \label{eq:final_time_complexity}
\end{align}

Because the global bounds $\epsilon_{max}$ and the amplification factor $a$ are determined by fixed algorithm configurations, the term marked by the underbrace is a constant. Thus, the expected execution time only grows linearly to the sampling size $s$.
\end{proof}

\clearpage

\section{Quantum Equivalence Checking ($\mathtt{QEC}$) Algorithm} \label{apx:QEC}

Please see Algorithm~\ref{alg:QEC}.

\LinesNumbered
\begin{algorithm}[tb]
\SetKwInOut{Input}{input}
\SetKwInOut{Output}{output}
\SetKwComment{Comment}{$\quad //$}{}
\Input{ $\ket{\phi}$: input state; \quad $n$: number of qubits; \quad $U_1$: Alice's circuit; \\ $U_2$: Bob's circuit; \quad
$m$: number of measurements; \quad $s$: number of samplings}
\Output{the estimation of $\widehat{I_{d,m}}\left((U_1 \otimes U_2) \ket{\phi}\right)$}
\BlankLine

$d \longleftarrow 2^{n}$ \;

\For{$t = 1$ \KwTo $s$}{
    $\ket{\psi} \longleftarrow (U_1 \otimes U_2) \ket{\phi}$ \;
    
    randomly select $r \in \{0,1\}$ and randomly select $i \in \{1, 2, \ldots, m\}$ \;

    \If{$r=0$}{
        $a \longleftarrow \widetilde{(A_i^1)_{\ket{\psi}}}$ \; 
        $b \longleftarrow \widetilde{(B_i^1)_{\ket{\psi}}}$ \;
        $X_t \longleftarrow 2 \cdot \alpha_{(a - b \mod d)}$ \;
    }
    \Else{
        $a \longleftarrow \widetilde{(A_{i+1}^1)_{\ket{\psi}}}$ \;
        $b \longleftarrow \widetilde{(B_i^1)_{\ket{\psi}}}$ \;
        $X_t \longleftarrow 2 \cdot \alpha_{(b - a \mod d)}$ \;
    }
}
\Return{$\frac{1}{s} \sum_{t = 1}^{s} X_t$}
    
\caption{Quantum Equivalence Checking ($\mathtt{QEC}$) ~\cite{SW22}}
\label{alg:QEC}
\end{algorithm}

\section{Early Termination of Algorithm~\ref{alg:QSAT}} \label{apx:EarlyTermination}

Another way to proactively avoid the worst case is to terminate Algorithm~\ref{alg:QSAT} earlier. If Algorithm~\ref{alg:QSAT} is terminated in $t$-th iteration, we could actually predict the result using the current $\mu_t$ based on its \emph{Bayes factor}~\cite{KR95}, as formulated in Definition~\ref{def:BayesFactor}.

\begin{definition} [Bayes Factor for $\mu_t$] \label{def:BayesFactor}
Let $H_0$ be the null hypothesis ($f$ has no solution) and $H_1$ be the alternative hypothesis ($f$ has solutions). For the $\mu_t$ obtained in the $t$-th iteration of Algorithm~\ref{alg:QSAT}, its Bayes factor, denoted by $\mathtt{BF}_{1,0}(\mu_t)$, is defined as $\mathtt{BF}_{1,0}(\mu_t) = \frac{\mathbb{P}(\mu_t \mid H_1)}{\mathbb{P}(\mu_t \mid H_0)}$, where $\mathbb{P}(\mu_t \mid H_K) = \sqrt{\frac{s}{2\pi\sigma^2(X_{d,m,K})}} \exp\left( -\frac{s(\mu_t - \widehat{I_{d,m,K}})^2}{2\sigma^2(X_{d,m,K})} \right)$.
\end{definition}

Basically, $\mathtt{BF}_{1,0}(\mu_t)$ provides quantitative evidence (strength) to support the alternative hypothesis $H_1$ (i.e., $f$ has solutions) over the null hypothesis $H_0$ (i.e., $f$ has no solution). To interpret the value of $\mathtt{BF}_{1,0}(\mu_t)$, we adopt the scale proposed by \cite{KR95}. A value of $\mathtt{BF}_{1,0}(\mu_t) \in (1, 3]$ is considered weak evidence, while $\mathtt{BF}_{1,0}(\mu_t) \in (3, 20]$ and $\mathtt{BF}_{1,0}(\mu_t) \in (20, 150]$ represent positive and strong evidence for $H_1$, respectively. For values $\mathtt{BF}_{1,0}(\mu_t) < 1$, the same scale is applied to the reciprocal $\mathtt{BF}_{0,1}(\mu_t) = 1/\mathtt{BF}_{1,0}(\mu_t)$ to support $H_0$ over $H_1$. Our experiments (c.f.~Section~\ref{sec:Evaluation}) show that Bayes factor is effective to avoid the worst case.



\zhe{The above is a bit confusing because $H_1$ is for the case of ``has solutions'', ie., the number of solutions could be 1, 2, 3,... But then we can't use a single $K$ in $\mathbb{P}(\mu_t \mid H_K)$. This makes it look like $H_1$ is the case of ``has 1 solution.'' Perhaps we can just simplify it to say $H_1$ represents exactly one solution, and then say that calculating the Bayes factor for $H_1$ versus  $H_0$ provides a conservative lower bound on the evidence. If the Bayes factor strongly supports $K = 1$ over $K = 0$, it will support $K > 1$ even more strongly.}

\end{document}